\documentclass[11pt]{article}
\usepackage[margin=1in]{geometry}

\usepackage{framed}
\usepackage{microtype}
\usepackage{enumerate}
\usepackage{comment}

\usepackage{relsize}
\usepackage{graphicx} 
\usepackage{sidecap} 
\usepackage{amsmath,amsfonts,amssymb,amsthm}
\usepackage{tikz}
\usepackage{nicefrac}
\usepackage{paralist}
\usepackage[ruled,vlined,linesnumbered]{algorithm2e}

\usetikzlibrary{shapes, arrows}

\usepackage{xcolor}
\usepackage[singlelinecheck=off,justification=raggedright,labelfont=bf]{caption}
\usepackage{floatrow}

\newfloatcommand{capbtabbox}{table}[][\FBwidth]

\usepackage{bbm}
\usepackage{mathtools}
\usepackage{color}
\usepackage[colorlinks,citecolor=green!60!black,linkcolor=red!60!black]{hyperref}
\theoremstyle{plain}

\newtheorem{theorem}{Theorem}
\newtheorem{corollary}[theorem]{Corollary}
\newtheorem{claim}{Claim}
\newtheorem{lemma}[theorem]{Lemma}
\newtheorem{example}[theorem]{Example}
\newtheorem{definition}[theorem]{Definition}
\newtheorem*{reduction}{Reduction}

\newcommand{\ProblemFormat}[1]{{\sc #1}}
\newcommand{\ProblemName}[1]{\ProblemFormat{#1}\xspace}

\newcommand{\probowakmedian}[0]{\ProblemName{OWA $k$-median}}
\newcommand{\metricprobowakmedian}[0]{\ProblemName{Metric OWA $k$-median}}

\newcommand{\probft}[0]{\ProblemName{Fault Tolerant}}
\newcommand{\probftkmed}[0]{\ProblemName{\probft\xspace\probkmedian}}
\newcommand{\probftkmedmulti}[0]{\ProblemName{\probftkmed with Clients Multiplicities}}
\newcommand{\probmetricftkmedmulti}[0]{\ProblemName{Metric \probftkmedmulti}}

\newcommand{\wmin}{{w\mathrm{_{dis}}}}

\newcommand{\owakmed}{{\mathrm{OWA\text{-}k\text{-}med}}}

\newcommand{\ftkmedmulti}{{\mathrm{FT\text{-}k\text{-}med\text{-}multi}}}

\newcommand{\probPAV}[0]{\ProblemName{Proportional Approval Voting}}
\newcommand{\probkmedian}[0]{\ProblemName{$k$-median}}

\newcommand{\probharkmedian}[0]{\ProblemName{Harmonic $k$-median}}

\newcommand{\probrftkmedian}[0]{\ProblemName{$r$-Fault Tolerant $k$-median}}
\newcommand{\probftkmedian}[0]{\ProblemName{Fault Tolerant $k$-median}}

\newcommand{\probkmeans}[0]{\ProblemName{$k$-means}}
\newcommand{\probkcenter}[0]{\ProblemName{$k$-center}}
\newcommand{\probsetcover}[0]{\ProblemName{Set Cover}}

\newcommand{\cost}{{\mathrm{cost}}}
\newcommand{\cov}{\mathrm{cov}}

\newcommand{\OPT}{{\mathrm{OPT}}}
\newcommand{\OPTLP}{{\mathrm{OPT}}^{\mathrm{LP}}}

\newcommand{\ratiohkm}{2.3589\xspace}
\newcommand{\xlp}{x^*}
\newcommand{\ylp}{y^*}
\newcommand{\xip}{x^\OPT}
\newcommand{\yip}{y^\OPT}
\newcommand{\lp}{linear program (\ref{eq:lp_min}--\ref{eq:lp_yixij})\xspace}

\newcommand{\prob}[1]{\mathrm{Pr}[#1]}
\newcommand{\probBigPar}[1]{\mathrm{Pr}\left[#1\right]}
\newcommand{\expected}[1]{\mathbb{E}[#1]}
\newcommand{\expectedBigPar}[1]{\mathbb{E}\left[#1\right]}

\newcommand{\Sum}{\ensuremath{\mathlarger{\sum}}}

\newcommand{\Dcal}{\ensuremath{\mathcal{D}}\xspace}
\newcommand{\Fcal}{\ensuremath{\mathcal{F}}\xspace}

\newcommand{\np}{{\mathsf{NP}}}
\newcommand{\p}{{\mathsf{P}}}
\newcommand{\calA}{\mathcal{A}}
\newcommand{\calS}{\mathcal{S}}

\newcommand{\calM}{\mathcal{M}}

\newcommand{\naturals}{{{\mathbb{N}}}}
\newcommand{\reals}{{{\mathbb{R}}}}

\newcommand{\har}{{{{\mathrm{har}}}}}
\newcommand{\geom}{{{{\mathrm{geom}}}}}
\newcommand{\sub}{{{{\mathrm{sub}}}}}
\newcommand{\submax}{{{{\mathrm{submax}}}}}

\allowdisplaybreaks

\begin{document}

\title{Proportional Approval Voting, Harmonic k-median,\\ and Negative Association}

\date{}


\author{
Jaros{\l}aw Byrka\thanks{University of Wroc{\l}aw, Wroc{\l}aw, Poland, \texttt{jby@cs.uni.wroc.pl}.}
\and
Piotr Skowron\thanks{University of Warsaw, Warsaw, Poland, \texttt{p.skowron@mimuw.edu.pl}.}
\and
Krzysztof Sornat\thanks{University of Wroc{\l}aw, Wroc{\l}aw, Poland, \texttt{krzysztof.sornat@cs.uni.wroc.pl}.}
}



\maketitle

\begin{abstract}
We study a generic framework that provides a unified view on two important classes of problems: (i) extensions of the k-median problem where clients are interested in having multiple facilities in their vicinity (e.g., due to the fact that, with some small probability, the closest facility might be malfunctioning and so might not be available for using), and (ii) finding winners according to some appealing multiwinner election rules, i.e., election system aimed for choosing representatives bodies, such as parliaments, based on preferences of a population of voters over individual candidates. Each problem in our framework is associated with a vector of weights: we show that the approximability of the problem depends on structural properties of these vectors. We specifically focus on the harmonic sequence of weights, since it results in particularly appealing properties of the considered problem. In particular, the objective function interpreted in a multiwinner election setup reflects to the well-known Proportional Approval Voting (PAV) rule.   

Our main result is that, due to the specific (harmonic) structure of weights, the problem allows constant factor approximation. This is surprising since the problem can be interpreted as a variant of the $k$-median problem where we do not assume that the connection costs satisfy the triangle inequality. To the best of our knowledge this is the first constant factor approximation algorithm for a variant of $k$-median that does not require this assumption.
The algorithm we propose is based on dependent rounding [Srinivasan, FOCS'01] applied to the solution of a natural LP-relaxation of the problem. The rounding process is well known to produce distributions over integral solutions satisfying \emph{Negative Correlation} (NC), which is usually sufficient for the analysis of approximation guarantees offered by rounding procedures. In our analysis, however, we need to use the fact that the carefully implemented rounding process
satisfies a stronger property, called \emph{Negative Association} (NA), which allows us to apply standard concentration bounds for \emph{conditional} random variables. 
\end{abstract}

\newpage
\section{Introduction}

This paper considers a general unified framework for two classes of problems: (i) extensions of the k-median problem where clients care about having multiple facilities in their vicinity, and (ii) finding winning committees according to a number of well-known, but hard-to-compute multiwinner election systems\footnote{We note that multiwinner election rules have many applications beyond the political domain---such applications include finding a set of results a search engine should display~\cite{DworkKNS01}, recommending a set of products a company should offer to its customers~\cite{LuB11,LuB15}, allocating shared resources among agents~\cite{SkowronFL16, mon:j:monroe}, solving variants of segmentation problems~\cite{KleinbergPR04}, or even improving genetic algorithms~\cite{FaliszewskiSSS16}.}. Let us first formalize our framework; we will discuss motivation and explain the relation to $k$-median and to multiwinner elections later on.

For a natural number $t \in \naturals$, by $[t]$ we denote the set $\{1, \ldots, t\}$.
Let $\Fcal= \{F_1, \ldots, F_m\}$ be the set of $m$ facilities and let $\Dcal = \{D_1, \ldots, D_n\}$ be the set of $n$ clients (demands). The goal is to pick a set of $k$ facilities that altogether are most satisfying for the clients. Different clients can have different preferences over individual facilities---by $c_{i, j}$ we denote the \emph{cost} that client $D_j$ suffers when using facility $F_i$ (this can be, e.g., the communication cost of client $D_j$ to facility $F_i$, or a value quantifying the level of personal dissatisfaction of $D_j$ from $F_i$). Following Yager~\cite{Yager88}, we use ordered weighted average (OWA) operators to define the cost of a client for a bundle of $k$ facilities $C$. Formally, let $w = \big(w_1, \ldots, w_k \big)$ be a non-increasing vector of $k$ weights. We define the $w$-cost of a client $D_j$ for a size-$k$ set of facilities $C$ as $w(C, j) = \sum_{i=1}^{k} w_i c_{i}^\rightarrow(C,j)$, where $c^\rightarrow(C,j) = (c_{1}^\rightarrow(C,j), \ldots, c_{k}^\rightarrow(C,j)) = \mathrm{sort}_\mathrm{ASC}\left(\big\{c_{i, j} \colon F_i \in C\big\} \right)$ is a non-decreasing permutation of the costs of client $D_j$ for the facilities from $C$. Informally speaking, the highest weight is applied to the lowest cost, the second highest weight to the second lowest cost, etc. In this paper we study the following computational problem.
\begin{definition}[\probowakmedian]
In \probowakmedian we are given a set $\Dcal = \{D_1, \ldots, D_n\}$ of clients, a set $\Fcal= \{F_1, \ldots, F_m\}$ of facilities, a collection of clients' costs $\big(c_{i, j}\big)_{i \in [m], j \in [n]}$, a positive integer $k$ ($k \leq m$), and a vector of $k$ non-increasing weights $w = \big(w_1, \ldots, w_k \big)$. The task is to compute a subset $C$ of $\Fcal$ that minimizes the value
\begin{equation*}
w(C) = \sum_{j=1}^{n} w(C, j) = \sum_{j=1}^{n} \sum_{i=1}^{k} w_i c_{i}^\rightarrow(C,j) \text{.}
\end{equation*}
\end{definition}

Note that \probowakmedian with weights $(1, 0, 0, \ldots, 0)$ is the \probkmedian problem.
Sometimes the costs represent distances between clients and facilities. Formally, this means that there exists a metric space $\calM$ with a distance function $d\colon \calM \times \calM \to \reals_{\geq 0}$, where each client and each facility can be associated with a point in $\calM$ so that for each $F_i \in \Fcal$ and each $D_j \in \Dcal$ we have $d(i, j) = c_{i,j}$. When this is the case, we say that the \emph{costs satisfy the triangle inequality}, and use the terms ``costs'' and ``distance'' interchangeably. Then, we use the prefix $\textsc{Metric}$ for the names of our problems. E.g., by \metricprobowakmedian we denote the variant of \probowakmedian where the costs satisfy the triangle inequality.

We are specifically interested in the following two sequences of weights:
\begin{enumerate}[(1)]
\vspace{-0.5em}
\item \textbf{harmonic:} $w_{\har} = \big(1, \nicefrac{1}{2}, \nicefrac{1}{3}, \ldots, \nicefrac{1}{k} \big)$. By \probharkmedian we denote the \probowakmedian problem with the harmonic vector of weights.
\item \textbf{$p$-geometric:} $w_{\geom} = \big(1, p, p^2, \ldots, p^{k-1} \big)$, for some $p < 1$.
\end{enumerate}
The two aforementioned sequences of weights, $w_{\har}$ and $w_{\geom}$, have their natural interpretations, which we discuss later on (for instance, see Examples~\ref{ex:harmonic_weights}~and~\ref{ex:geometric_weights}). 

\subsection{Motivation}

In this subsection we discuss the applicability of the studied model in two settings.

\subsubsection*{Multiwinner Elections}
Different variants of the \probowakmedian problem are very closely related to the preference aggregation methods and multiwinner election rules studied in the computational social choice, in particular, and in AI, in general---we summarize this relation in Table~\ref{tab:relation_between_problems} and in Figure~\ref{fig:owakmedian_sets}. In particular, one can observe that each ``median'' problem is associated with a corresponding ``winner'' problem. Specifically, the \probkmedian problem is known in computational social choice as the Chamberlin--Courant rule. Let us now explain the differences between the winner (``election'') and the median (``facility location'') problems:
\begin{enumerate}
\item The election problems are usually formulated as maximization problems, where instead of (negative) costs we have (positive) utilities.
The two variants, the minimization (with costs) and the maximization (with utilities) have the same optimal solutions. Yet, there is a substantial difference in their approximability.

Approximating the minimization variant is usually much harder. For instance, consider the Chamberlin--Courant (CC) rule which is defined by using the sequence of weights $(1, 0, 0, \ldots, 0)$. In the maximization variant standard arguments can be used to prove that a greedy procedure yields the approximation ratio of $(1 - \nicefrac{1}{e})$. This stands in a sharp contrast to the case when the same rule is expressed as the minimization one; in such a case we cannot hope for virtually any approximation~\cite{SkowronFS15} (we will extend this result in Theorem~\ref{thm:cc_no_approximation}). Approximating the minimization variant is also more desired. E.g., a $\nicefrac{1}{2}$-approximation algorithm for (maximization) CC can effectively ignore half of the population of clients, whereas it was argued~\cite{SkowronFS15} that a $2$-approximation algorithm for the minimization (if existed) would be more powerful. 
In this paper we study the harder minimization variant, and give the first constant-factor approximation algorithm for the minimization OWA-Winner with the harmonic weights.

\item In facility location problems it is usually assumed that the costs satisfy the triangle inequality. This relates to the previous point: since the problem cannot be well approximated in the general setting, one needs to make additional assumptions. One of our main results is showing that there is a $k$-median problem (\probowakmedian with harmonic weights) that admits a constant-factor approximation without assuming that the costs satisfy the triangle inequality; this is the first known result of this kind.  

\end{enumerate}

{
\renewcommand{\arraystretch}{2}
\begin{table}[t]
\footnotesize
\begin{tabular}{l|l | p{8cm}}
$k$-median problem & election rule & comment\\
\hline
\probowakmedian       & OWA-Winner~\cite{SkowronFL16}        & Finding winners according to OWA-Winner rules is the maximization variant of \probowakmedian (utilities instead of costs). \\
                & Thiele methods~\cite{thi:j:approval-OWA}  & Thiele methods are OWA-Winner rules for 0/1 costs. \\
\probharkmedian & PAV~\cite{thi:j:approval-OWA}           & In PAV we assume  the 0/1 costs. So far, only the maximization variant was considered in the literature.\\
\probkmedian    & Chamberlin--Courant~\cite{cha-cou:j:cc} & In CC, usually some specific form of utilities is assumed---different utilities have been considered, but always in the maximization variant (utilities instead of costs).
\end{tabular}
\caption{The relation between the \probkmedian problems and the corresponding problems studied in AI, in particular in the computational social choice community.}\label{tab:relation_between_problems}
\end{table}
}
\begin{figure}[t]
\captionsetup{
 format=plain,
 margin={-1em,0.5em},
 justification=justified,
}
\floatbox[{\capbeside\thisfloatsetup{capbesideposition={right,top},capbesidewidth=0cm}}]{figure}
{\caption{The relation between the considered models. \probowakmedian is the most general model. \probPAV and \probharkmedian due to the use of harmonic weights can be viewed as natural extensions of the well known and commonly used D'Hondt method of apportionment~\cite{BrillLS17}.}\label{fig:owakmedian_sets}}
{ \includegraphics[width=0.45\textwidth]{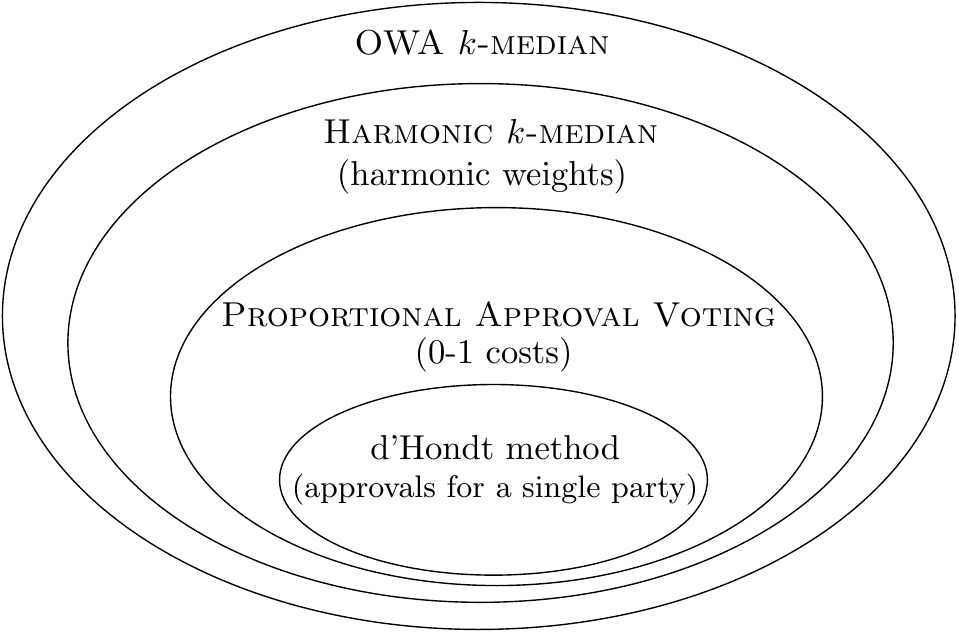}}
\end{figure}

The special case of \probharkmedian where each cost belongs to the binary set $\{0, 1\}$ is equivalent to finding winners according to \probPAV. The harmonic sequence $w_{\har} = (1, \nicefrac{1}{2}, \nicefrac{1}{3}, \ldots, \nicefrac{1}{k})$ is in a way exceptional: indeed, PAV can be viewed as an extension of the well known D'Hondt method of apportionment (used for electing parliaments in many contemporary democracies) to the case where the voters can vote for individual candidates rather than for political parties~\cite{BrillLS17}. Further, PAV satisfies several other appealing properties, such as extended justified representation~\cite{AzizBCEFW17}. This is one of the reasons why we are specifically interested in the harmonic weights. For more discussion on PAV and other approval-based rules, we refer the reader to the survey of Kilgour~\cite{kil-handbook}.

\subsubsection*{\probowakmedian as an Extension of \probkmedian}

Intuitively, our general formulation extends \probkmedian to scenarios where the clients not only use their most preferred facilities, but when there exists a more complex relation of ``using the facilities'' by the clients. Similar intuition is captured by  the \probft version of the \probkmedian problem introduced by Swamy and Shmoys~\cite{SwamyS08} and recently studied by Hajiaghayi et al.~\cite{HajiaghayiHLLS16}. There, the idea is that the facilities can be malfunctioning, and to increase the resilience to their failures each client needs to be connected to several of them.

\begin{definition}[\probftkmed]
In \probftkmedian problem we are given the same input as in \probkmedian, and additionally, for each client $D_j$ we are given a natural number $r_j \geq 1$, called \emph{the connectivity requirement}. The cost of a client $D_j$ is the sum of its costs for the $r_j$ closest open facilities. Similarly as in \probkmedian, we aim at choosing at most $k$ facilities so that the sum of the costs is minimized.
\end{definition}

When the values $\big(r_j\big)_{j \in [n]}$ are all the same, i.e., if $r_j = r$ for all $j$, then \probftkmed is called \probrftkmedian and it can be expressed as \probowakmedian for the weight vector $w$ with $r$ ones followed by $k-r$ zeros. Yet, in the typical setting of \probkmedian problems one additionally assumes that the costs between clients and facilities behave like distances, i.e., that they satisfy the triangle inequality. Indeed, the $(2.675+\epsilon)$-approximation algorithm for \probkmedian~\cite{ByrkaPRST17}, the $93$-approximation algorithm for \probftkmedian~\cite{HajiaghayiHLLS16}, the $2$-approximation algorithm for \probkcenter~\cite{HochbaumS85}, and the $6.357$-approximation algorithm for \probkmeans~\cite{AhmadianNSW17}, they all use triangle inequalities. Moreover it can be shown by straightforward reductions from the \probsetcover problem that there are no constant factor approximation algorithms for all these settings with general (non-metric) connection costs unless $\p = \np$.

Using harmonic or geometric OWA weights is also well-justified in case of facility location problems, as illustrated by the following examples.

\begin{example}[Harmonic weights: proportionality]\label{ex:harmonic_weights}
Assume there are $\ell \leq k$ cities, and for $i \in [\ell]$ let $N_i$ denote the set of clients who live in the $i$-th city. For the sake of simplicity, let us assume that $k\cdot |N_i|$ is divisible by $n$. Further, assume that the cost of traveling between any two points within a single city is negligible (equal to zero), and that the cost of traveling between different cities is equal to one. Our goal is to decide in which cities the $k$ facilities should be opened; naturally, we set the cost of a client for a facility opened in the same city to zero, and---in another city---to one. Let us consider \probowakmedian with the harmonic sequence of weights $w_{\har}$. Let $n_i$ denote the number of facilities opened in the $i$-th city in the optimal solution. We will show that for each $i$ we have $n_i = \frac{k|N_i|}{n}$, i.e., that the number of facilities opened in each city is proportional to its population. Towards a contradiction assume there are two cities, $i$ and $j$, with $n_i \geq \frac{k|N_i|}{n} + 1$ and $n_j \leq \frac{k|N_j|}{n} - 1$. By closing one facility in the $i$-th city and opening one in the $j$-th city, we decrease the total cost by at least:
\begin{align*}
|N_j| \cdot w_{n_j + 1} - |N_i| \cdot w_{n_i} = \frac{|N_i|}{n_j + 1} - \frac{|N_i|}{n_i} > \frac{|N_j| n}{k |N_j|} - \frac{|N_i| n}{k |N_i|} = 0 \text{.}
\end{align*}
Since, we decreased the cost of the clients, this could not be an optimal solution. As a result we see that indeed for each $i$ we have $n_i = \frac{k|N_i|}{n}$.
\end{example}

\begin{example}[Geometric weights: probabilities of failures]\label{ex:geometric_weights}
Assume that we want to select $k$ facilities and that each client will be using his or her favorite facility only. Yet, when a client wants to use a facility, it can be malfunctioning with some probability $p$; in such a case the client goes to her second most preferred facility; if the second facility is not working properly, the client goes to the third one, etc. Thus, a client uses her most preferred facility with probability $1-p$, her second most preferred facility with probability $p(1-p)$, the third one with probability $p^2(1-p)$, etc. As a result, the expected cost of a client $D_j$ for the bundle of $k$ facilities $C$ is equal to $w(C, j)$ for the weight vector $w = \big(1-p, (1-p)p, \ldots, (1-p)p^{k-1} \big)$. Finding a set of facilities, that minimize the expected cost of all clients is equivalent to solving \probowakmedian for the $p$-geometric sequence of weights (in fact, the sequence that we use is a $p$-geometric sequence multiplied by $(1-p)$, yet multiplication of the weight vector by a constant does not influence the structure of the optimal solutions).
\end{example}

\subsection{Our Results and Techniques}

Our main result is showing, that there exists a \ratiohkm-approximation algorithm for \probharkmedian for general connection costs (not assuming triangle inequalities). This is in contrast to the innaproximability of most clustering settings with general connection costs.

Our algorithm is based on dependent rounding of a solution to a natural linear program (LP) relaxation of the problem. We use the \emph{dependent rounding} (DR) studied by Srinivasan et al.~\cite{Srinivasan01,GandhiKPS06}, which transforms in a randomized way a fractional vector into an integral one. The sum-preservation property of DR ensures that exactly $k$ facilities are opened. 

DR satisfies, what is well known as \emph{negative correlation (NC)}---intuitively, this implies that the sums of subsets of random variables describing the outcome are more centered around their expected values than if the fractional variables were rounded independently. More precisely, negative correlation allows one to use standard concentration bounds such as the Chernoff-Hoeffding bound. Yet, interestingly, we find out that NC is not sufficient for our analysis in which we need a conditional variant of the concentration bound.    
The property that is sufficient for conditional bounds is
\emph{negative association} (NA)~\cite{negativeAssociation}.
In fact its special case that we call \emph{binary negative association} (BNA), is sufficient for our analysis. It captures the capability of reasoning about conditional probabilities. Thus, our work demonstrates how to apply the (B)NA property in the analysis of approximation algorithms based on DR. To the best of our knowledge, \probharkmedian is the first natural computational problem, where it is essential to use BNA in the analysis of the algorithm.

We additionally show that the 93-approximation algorithm of Hajiaghayi et al.~\cite{HajiaghayiHLLS16} can be extended to \probowakmedian (our technique is summarized in Section~\ref{sec:reduction_to_fault_tol})---this time we additionally need to assume that the costs satisfy the triangle inequality. Indeed, without this assumption the problem is hard to approximate for a large class of weight vectors; for instance, for $p$-geometric sequences with $p < \nicefrac{1}{e}$ (Theorem~\ref{thm:metric_geom_no_approximation} and Corollary~\ref{cor:metric_geom_no_approximation}) or for sequences where there exists $\lambda \in (0, 1)$ such that clients care only about the $\lambda$-fraction of opened facilities (Theorem~\ref{thm:cc_no_approximation}). Due to space constraints the formulation and the discussion on these hardness results are redelegated to Appendix~\ref{sec:hardness_general}. 

For the paper to be self-contained, in Appendix~\ref{sec:rounding_proc} we discuss in detail the process of dependent rounding (including a few illustrative examples); in particular, we provide an alternative proof that DR satisfies binary negative association. Our proof is more direct and shorter than the proofs known in the literature~\cite{KramerCR11}.

\section{\probharkmedian and \probPAV:\\a~$\ratiohkm$-approximation Algorithm}\label{sec:main_approximation}

In this section we demonstrate how to use the Binary Negative Association (BNA) property of Dependent Rounding (DR) to derive our main result---a randomized constant-factor approximation algorithm  for \probharkmedian. In Appendix~\ref{sec:rounding_proc} we provide a detailed discussion on DR and BNA, including a proof that DR satisfies BNA, and several examples.

\begin{theorem}\label{thm:alg_for_kmedian}
There exists a polynomial time randomized algorithm for \probharkmedian that gives $\ratiohkm$-approximation in expectation.
\end{theorem}

\begin{corollary}
There exists a polynomial time randomized algorithm for the minimization \probPAV that gives $\ratiohkm$-approximation in expectation.
\end{corollary}

In the remainder of this section we will prove the statement of Theorem~\ref{thm:alg_for_kmedian}.
Consider the following \lp that is a relaxation of a natural ILP for \probharkmedian.\\
\begin{minipage}[c]{0.43\columnwidth}
 \hspace{-1cm}
 \begin{align}
  \text{min} \quad \sum_{j=1}^{n} \sum_{\ell=1}^{k} \sum_{i=1}^{m} & \: w_\ell \cdot x_{ij}^\ell \cdot c_{ij}\label{eq:lp_min} &\\
  \sum_{i=1}^m y_i &= k\label{eq:lp_open_k}&
 \end{align}
\end{minipage}
\begin{minipage}[c]{0.56\columnwidth}
 \begin{align}
 \hspace{0.2\columnwidth} \sum_{\ell=1}^{k} x_{ij}^{\ell} &\le y_i \quad\; \forall i \in [m], \; j \in [n] \label{eq:lp_assigning}\\
 \sum_{i=1}^{m} x_{ij}^\ell &\ge 1 \quad\;\: \forall j \in [n], \; \ell \in [k] \label{eq:lp_served}
 \end{align}
\end{minipage}
 \vspace{-0.2cm}
 \begin{equation}
  \hspace{0.43\columnwidth} y_i, x_{ij}^\ell \in [0,1] \qquad \forall i \in [m], \; j \in [n], \; \ell \in [k]\label{eq:lp_yixij}
 \end{equation}

The intuitive meaning of the variables and constraints of the above LP is as follows. Variable $y_i$ denotes how much facility $F_i$ is opened. Integral values $1$ and $0$ correspond to, respectively, opening and not opening the $i$-th facility. Constraint \eqref{eq:lp_open_k} encodes opening exactly $k$ facilities. Each client $D_j \in \Dcal$ has to be assigned to each among $k$ opened facilities with different weights. For that we copy each client $k$ times: the $\ell$-th copy of a client $D_j$ is assigned to the $\ell$-th closest to $D_j$ open facility. Variable $x_{ij}^\ell$ denotes how much the $\ell$-th copy of $D_j$ is assigned to facility $F_i$. In an integral solution we have $x_{ij}^\ell \in \{0,1\}$, which means that the $\ell$-th copy of a client can be either assigned or not to the respective facility. The objective function \eqref{eq:lp_min} encodes the cost of assigning all copies of all clients to the opened facilities, applying proper weights. Constraint \eqref{eq:lp_assigning} prevents an assignment of a copy of a client to a not-opened part of a facility. In an integer solution it also forces assigning different copies of a client to different facilities. Observe that, due to non-increasing weights $w_\ell$, the objective \eqref{eq:lp_min} is smaller if an $\ell'$-th copy of a client is assigned to a closer facility than an $\ell''$-th copy, whenever $\ell'<\ell''$. Constraint \eqref{eq:lp_served} ensures that each copy of a client is served by some facility.

Just like in most facility location settings it is crucial to select the facilities to open, and the later assignment of clients to facilities can be done optimally by a simple greedy procedure. We propose to select the set of facilities in a randomized way by applying the DR procedure to the $y$ vector from an optimal fractional solution to \lp. This turns out to be a surprisingly effective methodology for \probharkmedian.

\subsection{Analysis of the Algorithm}

Let $\OPTLP$ be the value of an optimal solution $(\xlp,\ylp)$ to the \lp. Let $\OPT$ be the value of an optimal solution $(\xip,\yip)$ for \probharkmedian. Easily we can see that $(\xip,\yip)$ is a feasible solution to the \lp, so $\OPTLP \leq \OPT$. Let $Y=(Y_1,\dots,Y_m)$ be the random solution obtained by applying the DR procedure described in Appendix~\ref{sec:rounding_proc} to the vector $\ylp$. Recall that DR preserves the sum of entries (see Appendix~\ref{sec:rounding_proc}), hence we have exactly $k$ facilities opened. It is straightforward to assign clients to the open facilities, so the variables $X = (X_{ij}^\ell)_{j \in [n], i \in [m], \ell \in [k]}$ are easily determined. 

We will show that $\expected{\cost(Y)} \leq \ratiohkm \cdot \OPTLP$. In fact, we will show that $\expected{\cost_j(Y)} \leq \ratiohkm \cdot \OPTLP_j$, where the subindex $j$ extracts the cost of assigning client $D_j$ to the facilities in the solution returned by the algorithm. 
In our analysis we focus on a single client $D_j \in \Dcal$. Next, we reorder the facilities $\{F_1,F_2,\dots,F_m\}$ in the non-decreasing order of their connection costs to $D_j$ (i.e., in the non-decreasing order of $c_{ij}$). Thus, from now on, facility $F_i$ is the $i$-th closest facility to client $D_j$; ties are resolved in an arbitrary but fixed way. 

\begin{figure}[t]
  \includegraphics[trim=0 0 0 0, clip, scale=0.78]{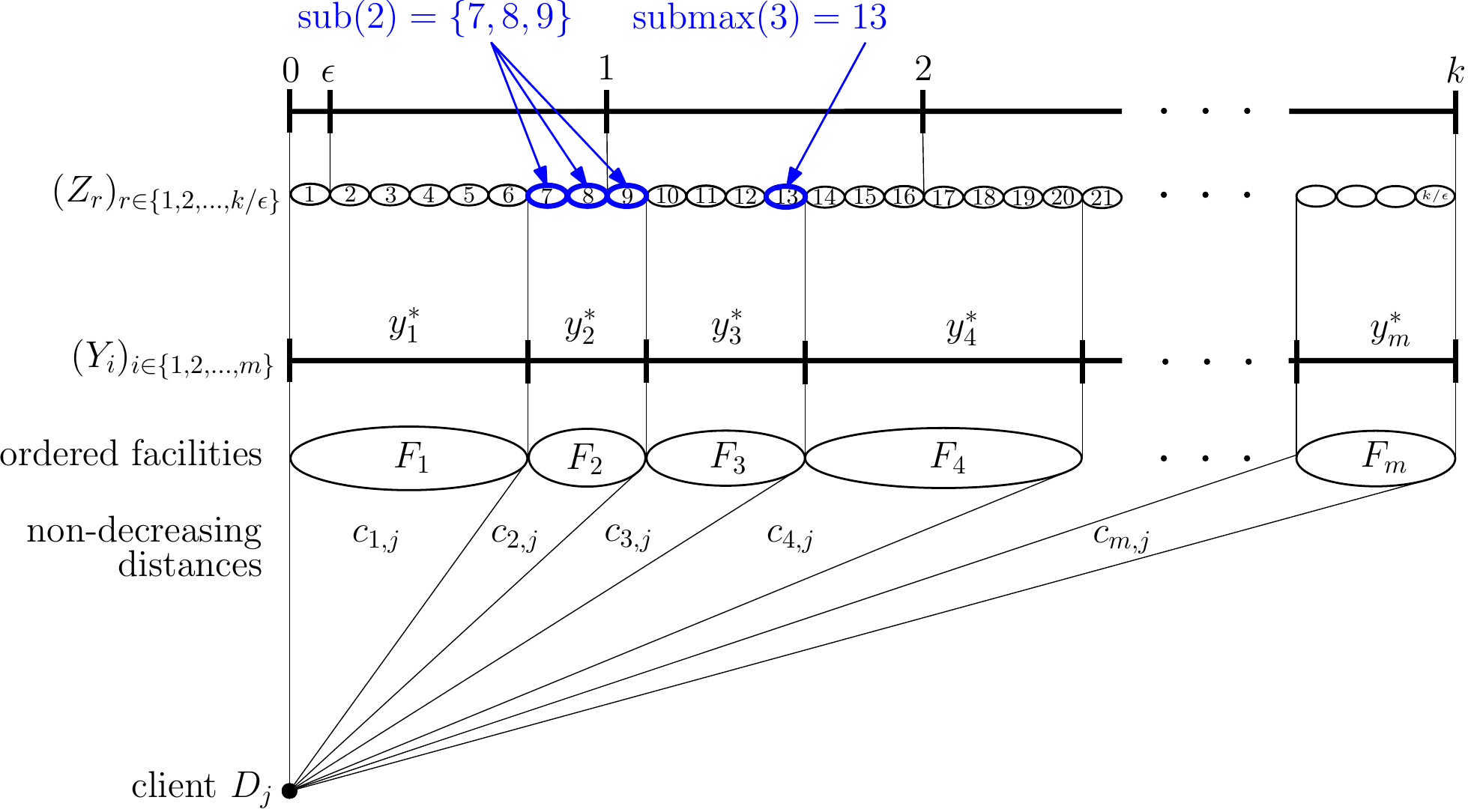}
  \caption{Ordering of the facilities by $c_{i,j}$ for the chosen client $D_j$. Definitions of the variables $Y_i$, $Z_r$ and of the indices $\sub(i)$ and $\submax(i)$.}\label{fig:ordering}
\end{figure} 

The ordering of the facilities is depicted in Figure~\ref{fig:ordering}, which also includes information about the fractional opening of facilities in $\ylp$, i.e., facility $F_i$ is represented by an interval of length $y_i^*$. The total length of all intervals equals $k$. Next, we subdivide each interval into a set of (small) $\epsilon$-size pieces (called $\epsilon$-subintervals); $\epsilon$ is selected so that $\nicefrac{1}{\epsilon}$, and $\nicefrac{y_i^*}{\epsilon}$ for each $i$, are integers. 
Note that the values $y_i^*$, which originate from the solution returned by an LP solver, are rational numbers. The subdivision of $[0,k]$ into $\epsilon$-subintervals is shown in Figure~\ref{fig:ordering} on the "$(Z_r)_{r \in \{1,2,\dots,\nicefrac{k}{\epsilon}\}}$" level.

The idea behind introducing the $\epsilon$-subintervals is the following. Although computationally the algorithm applies DR to the $\ylp$ variables, for the sake of the analysis we may think that the DR process is actually rounding $z$ variables corresponding to $\epsilon$-subinterval under the additional assumption that rounding within individual facilities is done before rounding between facilities.
Formally, we replace the vector $Y=(Y_1,Y_2,\dots,Y_m)$ by an equivalent vector of random variables $Z=(Z_1,Z_2,\dots,Z_{\nicefrac{k}{\epsilon}})$. Random variable $Z_r$ represents the $r$-th $\epsilon$-subinterval. We will use the following notation to describe the bundles of $\epsilon$-subintervals that correspond to particular facilities:
\begin{align}
&\submax(0) = 0 \qquad \text{and} \qquad  \submax(i) = \submax(i-1) + \frac{y_i^*}{\epsilon}, \label{def:submax} \\
&\sub(i) = \{\submax(i-1)+1, \dots, \submax(i) \}.\label{def:sub} 
\end{align}
Intuitively, $\sub(i)$ is the set of indexes $r$ such that $Z_r$ represents an interval belonging to the $i$-th facility. Examples for both definitions are shown in Figure~\ref{fig:ordering} in the upper level. Formally, the random variables $Z_r$ are defined so that:
\begin{eqnarray}
Y_i = \sum_{r \in \sub(i)} Z_r \qquad \text{and} \qquad Y_i = 1 \implies \exists ! \; r \in \sub(i) \quad Z_r = 1. \label{eq:defZr}
\end{eqnarray}

\noindent
For each $r \in \{1,2,\dots, \nicefrac{k}{\epsilon}\}$ we can write that:
\begin{equation}\label{eq:pr_zr1}
\prob{Z_r=1} = \prob{Z_r=1 \big| Y_{\sub^{-1}(r)}=1} \cdot \prob{Y_{\sub^{-1}(r)}=1} = \frac{\epsilon}{y_{\sub^{-1}(r)}^*} \cdot y_{\sub^{-1}(r)}^* = \epsilon 
\end{equation}
and $\prob{Z_r=0} = 1 - \epsilon$, hence $\expected{Z_r} = \epsilon$.
Also we have:
\begin{equation}\label{eq:pr_sumZr}
\probBigPar{Y_i = 1} = \probBigPar{\sum_{r \in \sub(i)} Z_r = 1} = \probBigPar{\bigvee_{r \in \sub(i)} Z_r = 1} = \sum_{r \in \sub(i)} \probBigPar{Z_r = 1}.
\end{equation}

\noindent
When $Y_i = 1$ its representative is chosen randomly among  $(Z_r)_{r \in \sub(i)}$ independently of the choices of representatives of other facilities. Therefore
\begin{equation}\label{eq:conditioning_Yi_Zr}
\forall_{i \in [m]} \quad \forall_{r \in \sub(i)} \quad \expectedBigPar{f(Y) \;|\; Y_i=1} = \expectedBigPar{f(Y) \;|\; Y_i=1 \;\wedge\; Z_r=1},
\end{equation}
for any function $f$ on vector $Y = (Y_1,Y_2, \dots, Y_m)$.

\noindent
Now we are ready to analyze the expected cost for any client $D_j \in \Dcal$.

\begin{eqnarray}
\expected{\cost_j(Y)} \hspace{-0.5cm} &\leq& \Sum_{i=1}^{m} \: \left( \expectedBigPar{\mathlarger{\frac{c_{ij}}{1+\sum_{i'=1}^{i-1} Y_{i'}}} \Bigg| Y_i = 1} \cdot \probBigPar{Y_i=1} \right)\nonumber\\
&\stackrel{\eqref{eq:pr_sumZr}}{=}& \Sum_{i=1}^{m} \: \left( c_{ij} \cdot \expectedBigPar{\mathlarger{\frac{1}{1+\sum_{i'=1}^{i-1} Y_{i'}}} \Bigg| Y_i = 1} \cdot \Sum_{r \in \sub(i)} \probBigPar{Z_r=1} \right) \nonumber\\
&=& \Sum_{i=1}^{m} \: \left( c_{ij} \cdot \Sum_{r \in \sub(i)} \expectedBigPar{\mathlarger{\frac{1}{1+\sum_{i'=1}^{i-1} Y_{i'}}} \Bigg| Y_i = 1} \cdot \probBigPar{Z_r=1} \right) \nonumber\\
&\stackrel{\eqref{eq:conditioning_Yi_Zr}}{=}& \Sum_{i=1}^{m} \left( c_{ij} \cdot \Sum_{r \in \sub(i)} \expectedBigPar{\mathlarger{\frac{1}{1+\sum_{i'=1}^{i-1} Y_{i'}}} \Bigg| Y_i = 1 \wedge Z_r = 1} \cdot \probBigPar{Z_r = 1} \right) \nonumber\\
&\stackrel{\eqref{eq:defZr},\eqref{eq:pr_zr1}}{=}& \Sum_{i=1}^{m} \left( \epsilon \cdot c_{ij} \cdot \Sum_{r \in \sub(i)} \expectedBigPar{\mathlarger{\frac{1}{1+\sum_{r'=1}^{\submax(i-1)} Z_{r'}}} \Bigg| Z_r = 1} \right) \nonumber\\
&\stackrel{\eqref{eq:defZr}}{=}& \Sum_{i=1}^{m} \left( \epsilon \cdot c_{ij} \cdot \Sum_{r \in \sub(i)} \expectedBigPar{\mathlarger{\frac{1}{1+\sum_{r'=1}^{r-1} Z_{r'}}} \Bigg| Z_r = 1} \right)\label{ineq:ecostjy}
\end{eqnarray}
W.l.o.g., assume that $\OPTLP_j > 0$. Hence the approximation ratio for any client $D_j$ is
\begin{eqnarray*}
 \frac{\expected{\cost_j(Y)}}{\OPTLP_j} \stackrel{\eqref{def:sub},\eqref{ineq:ecostjy}}{\leq} \frac{\Sum_{r=1}^{\nicefrac{k}{\epsilon}} \epsilon \cdot c_{\sub^{-1}(r),j} \cdot \expectedBigPar{\mathlarger{\frac{1}{1+\sum_{r'=1}^{r-1} Z_{r'}}} \bigg| Z_r = 1}}{\Sum_{r=1}^{\nicefrac{k}{\epsilon}} \epsilon \cdot c_{\sub^{-1}(r),j} \cdot \frac{1}{\lceil r\epsilon \rceil} } =
\end{eqnarray*}
note that $\sub^{-1}(r)$ is an index of a facility that contains $Z_r$. Now we convert the sum over facilities into a sum over unit intervals. A unit interval is represented as a sum of $\nicefrac{1}{\epsilon}$ many $\epsilon$-subintervals:
\begin{eqnarray*}
&=& \frac{\Sum_{\ell=1}^{k} \; \Sum_{r=\nicefrac{(\ell-1)}{\epsilon} + 1}^{\nicefrac{\ell}{\epsilon}} c_{\sub^{-1}(r),j} \cdot \expectedBigPar{\mathlarger{\frac{1}{1+\sum_{r'=1}^{r-1} Z_{r'}}} \bigg| Z_r = 1}}{\Sum_{\ell=1}^{k} \; \Sum_{r=\nicefrac{(\ell-1)}{\epsilon} + 1}^{\nicefrac{\ell}{\epsilon}} c_{\sub^{-1}(r),j} \cdot \frac{1}{\ell}} \leq \nonumber
\end{eqnarray*}
W.l.o.g., we can assume that first interval has non-zero costs: $\sum_{r=1}^{\nicefrac{1}{\epsilon}} c_{\sub^{-1}(r),j} > 0$, otherwise the LP pays $0$ and our algorithm pays $0$ in expectation on intervals from non-empty prefix of $(1,2,\dots,k)$. With this assumption we can take maximum over intervals:
\begin{eqnarray*}
&\stackrel{\text{Lemma \ref{lem:sai_over_sbi_leq_max_ai_over_bi}}}{\leq}& \max_{\ell \in [k]} \left( \frac{\Sum_{r=\nicefrac{(\ell-1)}{\epsilon} + 1}^{\nicefrac{\ell}{\epsilon}} c_{\sub^{-1}(r),j} \cdot \expectedBigPar{\mathlarger{\frac{1}{1+\sum_{r'=1}^{r-1} Z_{r'}}} \bigg| Z_r = 1}}{\Sum_{r=\nicefrac{(\ell-1)}{\epsilon} + 1}^{\nicefrac{\ell}{\epsilon}} c_{\sub^{-1}(r),j} \cdot \frac{1}{\ell}} \right) \leq \nonumber
\end{eqnarray*}
Costs $c_{\sub^{-1}(r),j}$ can be general and they could be hard to analyze. Therefore we would like to remove costs from the analysis. We will use Lemma \ref{lem:ca_over_cb_leq_ab} for which the technique of splitting variables $Y_i$ into $Z_r$ was needed. We are using the fact that the variables $Z_r$ have the same expected values; otherwise the coefficient in front of the expected value would be $c_{ij} \cdot y_i^*$, i.e., not monotonic. Thus
\begin{eqnarray}
&\stackrel{\text{Lemma \ref{lem:ca_over_cb_leq_ab}}}{\leq}& \max_{\ell \in [k]} \; \left( \epsilon \cdot \ell \cdot \Sum_{r=\nicefrac{(\ell-1)}{\epsilon} + 1}^{\nicefrac{\ell}{\epsilon}} \; \expectedBigPar{\mathlarger{\frac{1}{1+\sum_{r'=1}^{r-1} Z_{r'}}} \bigg| Z_r = 1} \right) \label{ineq:approx_ratio}.
\end{eqnarray}
Consider the expected value in the above expression for a fixed $r \in \{\nicefrac{(\ell-1)}{\epsilon} + 1,\dots, \nicefrac{\ell}{\epsilon} \}$:
\begin{eqnarray}\label{eq:Er}
E_r &=& \expectedBigPar{\mathlarger{\frac{1}{1+\sum_{r'=1}^{r-1} Z_{r'}}} \bigg| Z_r = 1} = \Sum_{t=1}^{k} \frac{1}{t} \probBigPar{\Sum_{r'=1}^{r-1} Z_{r'}=t-1 \bigg| Z_{r} = 1} = \nonumber\\
&=& \Sum_{t=1}^{\ell} \frac{1}{t} \probBigPar{\Sum_{r'=1}^{r-1} Z_{r'}=t-1 \bigg| Z_{r} = 1} + \Sum_{t=\ell+1}^{k} \frac{1}{t} \probBigPar{\Sum_{r'=1}^{r-1} Z_{r'}=t-1 \bigg| Z_{r} = 1}.
\end{eqnarray}
For $t \in \{1,2,\dots,\ell\}$ we consider the conditional probability in the above expression, denote it by $p_r(t-1)$, and analyze the corresponding cumulative distribution function $H_r(t-1)$:
\begin{eqnarray}
p_r(t-1) &=& \probBigPar{\Sum_{r'=1}^{r-1} Z_{r'}=t-1 \bigg| Z_{r} = 1}, \label{eq:pz}\\
H_r(t-1) &=& \probBigPar{\Sum_{r'=1}^{r-1} Z_{r'} \leq t-1 \bigg| Z_{r} = 1} = \sum_{t'=0}^{t-1} p_r(t'), \label{eq:Hr_def}
\end{eqnarray}
We continue the analysis of $E_r$:
\begin{eqnarray}
E_r &\stackrel{\eqref{eq:Er},\eqref{eq:pz}}{=}& \Sum_{t=1}^{\ell} \frac{1}{t} p_r(t-1) + \Sum_{t=\ell+1}^{k} \frac{1}{t} p_r(t-1)  \nonumber\\
&\stackrel{\eqref{eq:Hr_def}}{=} & H_r(0) + \Sum_{t=2}^{\ell} \frac{1}{t} \left( H_r(t-1) - H_r(t-2) \right) + \Sum_{t=\ell+1}^{k} \frac{1}{t} p_r(t-1)  \nonumber\\
&=& H_r(0) + \Sum_{t=2}^{\ell} \frac{1}{t} H_r(t-1) - \Sum_{t=2}^{\ell} \frac{1}{t} H_r(t-2) + \Sum_{t=\ell+1}^{k} \frac{1}{t} p_r(t-1)  \nonumber\\
&=& \Sum_{t=1}^{\ell} \frac{1}{t} H_r(t-1) - \Sum_{t=1}^{\ell-1} \frac{1}{t+1} H_r(t-1) +\Sum_{t=\ell+1}^{k} \frac{1}{t} p_r(t-1)  \nonumber\\
&=& \Sum_{t=1}^{\ell-1} \frac{1}{t} H_r(t-1) - \Sum_{t=1}^{\ell-1} \frac{1}{t+1} H_r(t-1) + \frac{1}{\ell} H_r(\ell-1) +\Sum_{t=l+1}^{k} \frac{1}{t} p_r(t-1) \nonumber\\
&\leq& \Sum_{t=1}^{\ell-1} \left(\frac{1}{t}-\frac{1}{t+1}\right) H_r(t-1) + \frac{1}{\ell} \left( H_r(\ell-1) +\Sum_{t=\ell+1}^{k} p_r(t-1) \right) \nonumber\\
&=& \Sum_{t=1}^{\ell-1} \frac{1}{t(t+1)} H_r(t-1) + \frac{1}{\ell} \left( H_r(\ell-1) +\Sum_{t=\ell+1}^{k} p_r(t-1) \right) \nonumber\\
&\leq& \Sum_{t=1}^{\ell-1} \frac{1}{t(t+1)} H_r(t-1) + \frac{1}{\ell}.\label{ineq:Er}
\end{eqnarray}

\begin{lemma}\label{lem:Hr_chernoff}
For any $\ell \in [k]$, $t \in [\ell-1]$ and $r \in \{\nicefrac{(\ell-1)}{\epsilon}+1,\nicefrac{(\ell-1)}{\epsilon}+2,\dots,\nicefrac{\ell}{\epsilon}\}$ we have
\[ H_r(t-1) \leq e^{-r \cdot \epsilon} \cdot \left(\frac{e \cdot r \cdot \epsilon}{t} \right)^t. \]
\end{lemma}
The proof of Lemma~\ref{lem:Hr_chernoff} combines the use of the BNA property of variables $\{Z_1,Z_2, \dots, Z_{\nicefrac{k}{\epsilon}}\}$ with applications of Chernoff-Hoeffding bounds. Due to the space constraints, the proof is moved to the Appendix~\ref{appen:omitted_:main_approximation}. In the end, we get the following bound on the approximation ratio.
\begin{lemma}\label{lem:final_apx_ratio}
For any $j \in [n]$ we have
 \[ \frac{\expected{\cost_j(Y)}}{\OPTLP_j} \leq \ratiohkm. \]
\end{lemma}

A proof uses inequalities~\eqref{ineq:approx_ratio}, \eqref{ineq:Er} as well as Lemma~\ref{lem:Hr_chernoff} with an upper bound derived by an integral of the function $f_t(x) = e^{-x}$. We made numerical calculation for $\ell \in \{1,2,\dots,88\}$ and for other case we used Stirling formula and Taylor series for $e^\ell$ to derive analytical upper bound. 
Full proof, including a plot of numericaly obtained values, is presented in the Appendix~\ref{appen:omitted_:main_approximation}.

\section{\probowakmedian with Costs Satisfying the Triangle Inequality}\label{sec:reduction_to_fault_tol}

In this section we construct an algorithm for \probowakmedian with costs satisfying the triangle inequality. Thus, the problem we address in this section is more general than \probharkmedian (i.e., the problem  we have considered in the previous section) in a sense that we allow for arbitrary non-increasing sequences of weights. On the other hand, it is less general in a sense that we require the costs to form a specific structure (a metric). 

In our approach we first adapt the algorithm of Hajiaghayi~et~al.~\cite{HajiaghayiHLLS16} for \probftkmed so that it applies to the following, slightly more general setting: for each client $D_j$ we introduce its multiplicity $m_j \in \mathbb{N}$---intuitively, this corresponds to cloning $D_j$ and co-locating all such clones in the same location as $D_j$. However, this will require a modification of the original algorithm for \probftkmed, since we want to allow the multiplicities $\{m_j\}_{D_j \in \mathcal{D}}$ to be exponential with respect to the size of the instance (otherwise, we could simply copy each client a sufficient number of times, and use the original algorithm of Hajiaghayi~et~al.). 

Next, we provide a reduction from \probowakmedian to such a generalization of \probftkmed. The resulting \probftkmedmulti problem can be cast as the following integer program:\\
 \begin{minipage}[c]{0.43\columnwidth}
 \hspace{-1cm}
 \begin{align}
 \text{min} \quad \sum_{j=1}^{n} &\sum_{i=1}^{m} \: m_j \cdot x_{ij} \cdot c_{ij} &\nonumber\\
  &\sum_{i=1}^m y_i = k &\nonumber
 \end{align}
\end{minipage}
\begin{minipage}[c]{0.56\columnwidth}
 \begin{align}
 \hspace{0.0\columnwidth} \sum_{i=1}^{m} x_{ij} &= r_j \qquad\;\: &\forall j \in [n]\nonumber\\
 x_{ij} &\le y_i \qquad\;\: &\forall i \in [m], \; j \in [n]\nonumber\\
 y_i, x_{ij} &\in \{0,1\} \quad &\forall i \in [m]\nonumber\\
 m_j &\in \naturals \qquad\;\: &\forall j \in [n] \nonumber
 \end{align}
\end{minipage}

\newpage
\begin{theorem}\label{thm:93-apx-for-multi}
 There is a polynomial-time $93$-approximation algorithm for \probmetricftkmedmulti.
\end{theorem}

Proof can be found in the Appendix~\ref{sec:apxD}.

\medskip

Consider reduction from \probowakmedian to \probftkmedmulti depicted on Figure~\ref{fig:reduction}.

\begin{figure}

\begin{framed}
 \begin{reduction}
Let us take an instance $I$ of \probowakmedian $\big(\Dcal ,\Fcal ,k ,w , \{c_{ij}\}_{F_i \in \Fcal, D_j \in \Dcal} \big)$ where $w_i = \frac{p_i}{q_i}, i \in [k]$ are rational numbers in the canonical form. We construct an instance $I'$ of \probftkmedmulti with the same set of facilities and the same number of facilities to open, $k$. Each client $D_j \in \Dcal$ is replaced with clients $D_{j,1}, D_{j,2}, \dots, D_{j,k}$ with requirements $1,2, \dots, k$, respectively. For $Q = \prod_{r=1}^{k} q_r$, the multiples of the clients are defined as follows:
\begin{itemize}
  \item $m_{j,\ell} = (w_\ell - w_{\ell+1}) \cdot Q$, for each $\ell \in [k-1]$, and
  \item $m_{j,k} = w_k \cdot Q$.
 \end{itemize}
 \end{reduction}
\end{framed}
\caption{Reduction from \probowakmedian to \probftkmedmulti.} \label{fig:reduction}
\end{figure}

\begin{lemma}\label{lem:costftkmi-gen}
Let $I$ be an instance of \probowakmedian, and let $I'$ be an instance of \probftkmedmulti constructed from $I$ through reduction from Figure~\ref{fig:reduction}. An $\alpha$-approximate solution to $I'$ is also an $\alpha$-approximate solution to $I$.
\end{lemma}

Proof can be found in the Appendix~\ref{sec:apxD}.



\begin{corollary}\label{cor:93_metricowakmedian}
There exists a $93$-approximation algorithm for \metricprobowakmedian that runs in polynomial time.
\end{corollary}

\section{Concluding Remarks and Open Questions}\label{sec:conclusions}

We have introduced a new family of $k$-median problems, called \probowakmedian, and we have shown that our problem with the harmonic sequence of weights allows for a constant factor approximation even for general (non-metric) costs. This algorithm applies to Proportional Approval Voting.
In the analysis of our approximation algorithm for \probharkmedian, we used the fact that the dependent rounding procedure satisfies Binary Negative Association. 

We showed that any \metricprobowakmedian can be approximated within a factor of $93$ via a reduction to \probftkmedmulti. We also obtained that \probowakmedian with $p$-geometric weights with $p < \nicefrac{1}{e}$ cannot be approximated without the assumption of the costs being metric. The status of the non-metric problem with $p$-geometric weights with $p > \nicefrac{1}{e}$ remains an intriguing open problem.


Using approximation and randomized algorithms for finding winners of elections requires some comment. First, the multiwinner election rules such as PAV have many applications in the voting theory, recommendation systems and in resource allocation. Using (randomized) approximation algorithms in such scenarios is clearly justified. However, even for other high-stake domains, such as political elections, the use of approximation algorithms is a promising direction. One approach is to view an approximation algorithm as a new, full-fledged voting rule (for more discussion on this, see the works of Caragiannis et al.~\cite{CaragiannisCFHKKPR12,CaragiannisKKP14}, Skowron et al.~\cite{SkowronFS15}, and Elkind et al.~\cite{ElkindFSS17}). 
In fact, the use of randomized algorithms in this context has been advocated in the literature as well---e.g., one can arrange an election where each participant is allowed to suggest a winning committee, and the best out of the suggested committees is selected; in such case the approximation guaranty of the algorithm corresponds to the quality of the outcome of elections (for a more detailed discussion see~\cite{SkowronFS15})~\footnote{Indeed, approximation algorithms for many election rules have been extensively studied in the literature. In the world of single-winner rules, there are already very good approximation algorithms known for
the Kemeny's rule~\cite{AilonCN08,CoppersmithFR10,Kenyon-MathieuS07}
and for the Dodgson's rule~\cite{McCabe-DanstedPS08,HomanH09,CaragiannisCFHKKPR12,FaliszewskiHH11,CaragiannisKKP14}. A hardness of approximation has been proven for the Young's rule~\cite{CaragiannisCFHKKPR12}. For the multiwinner case we know good (randomized) approximation algorithms for Minimax Approval Voting~\cite{CKSS17a}, Chamberlin--Courant rule~\cite{SkowronFS15}, Monroe rule~\cite{SkowronFS15}, or maximization variant of PAV~\cite{SkowronFL16}.}. Nonetheless, we think that it would be beneficial to learn whether our algorithm can be efficiently derandomized. 


\section*{Acknowledgments.}
We thank Ola Svensson and Aravind Srinivasan for their helpful and insightful comments.
\newline\newline
\noindent
J.~Byrka was supported by the National Science Centre, Poland, grant number 2015/18/E/ST6/00456.
P.~Skowron was supported by a Humboldt Research Fellowship for Postdoctoral Researchers.\\
K.~Sornat was supported by the National Science Centre, Poland, grant number 2015/17/N/ST6/03684.

\bibliography{main}

\appendix

\section{Dependent Rounding and Negative Association}\label{sec:rounding_proc}

Consider a vector of $m$ variables $(y_{i})_{i \in [m]}$, and let $y^*_{i}$ denote the initial value of the variable $y_{i}$. For simplicity we will assume that $0 \leq y^*_{i} \leq 1$ for each $i$, and that $k = \sum_{i \in [m]}y^*_{i}$ is an integer. A rounding procedure takes this vector of (fractional) variables as an input, and transforms it into a vector of 0/1 integers. We focus on a specific rounding procedure studied by Srinivasan~\cite{Srinivasan01} which we refer to as \emph{dependent rounding} (DR).

DR works in steps: in each step it selects two fractional variables, say $y_i$ and $y_j$, and changes the values of these variables to $y_i'$ and $y_j'$ so that $y_i' + y_j' = y_i + y_j$, and so that $y_i'$ or $y_j'$ is an integer. Thus, after each iteration at least one additional variable becomes an integer. The rounding procedure stops, when all variables are integers. 
In each step the randomization is involved: with some probability $p$ variable $y_i$ is rounded to an integer value, and with probability $1-p$ variable $y_j$ becomes an integer. The value of the probability $p$ is selected so as to preserve the expected value of each individual entry $y_i$. Clearly, if $y_i + y_j \geq 1$, then one of the variables is rounded to 1; otherwise, one of the variables is rounded to 0. For example, if $y_i = 0.4$ and $y_j = 0.8$, then with probability $0.25$ the values of the variables $y_i$ and $y_j$ change to, respectively, $1$ and $0.2$; and with probability $0.75$ they change to, respectively, $0.2$ and $1$. If $y_i = 0.3$ and $y_j = 0.2$, then with probability $0.4$ the values of the two variables change to, respectively, $0$ and $0.5$; and with probability $0.6$, to, respectively, $0.5$ and $0$.

Let $Y_i$ denote the random variable which returns one if $y_i$ is rounded to one after the whole rounding procedure, and zero, otherwise. It was shown~\cite{Srinivasan01} that the DR generates distributions of $Y_i$ which satisfy the following three properties:
\begin{description}
\item[Marginals.] $\prob{Y_i = 1} = y^*_{i}$,
\item[Sum Preservation.] $\prob{\sum_i Y_i = k} = 1$,
\item[Negative Correlation.] For each $S \subseteq [m]$ it holds that $\prob{\bigwedge_{i \in S} (Y_i = 1)} \leq \prod_{i \in S}\prob{Y_i = 1}$, and $\prob{\bigwedge_{i \in S} (Y_i = 0)} \leq \prod_{i \in S}\prob{Y_i = 0}$.
\end{description}  

These three properties are often used in the analysis of approximation algorithms based on dependent rounding for various optimization problems---see, e.g.,~\cite{GandhiKPS06}. In fact, DR satisfies an even stronger property than NC, called \emph{conditional negative association} (CNA)~\cite{KramerCR11}, yet, to the best of our knowledge, this property has never been used before for analyzing algorithms based on the DR procedure.

For two random variables, $X$ and $Y$, by $\cov[X, Y]$ we denote the covariance between $X$ and $Y$. Recall that $\cov[X, Y] = \expected{XY} - \expected{X} \cdot \expected{Y}$.

\begin{description}
\item[Negative Association~\cite{negativeAssociation}.] For each $S, Q \subseteq [m]$ with $S \cap Q = \emptyset$, $s = |S|$, and $q = |Q|$, and each two nondecreasing functions, $f\colon [0, 1]^s \to \reals$ and $g\colon [0, 1]^q \to \reals$, it holds that:
\begin{align*}
\cov\big[f(Y_i\colon i \in S), g(Y_i\colon i \in Q)\big] \leq 0.
\end{align*}

\item[Conditional Negative Association.] We say that the sequence of random variables $(Y_i)_{i \in [m]}$ satisfies the CNA property if the conditional variables $(Y_{[m] \setminus S} | Y_S = a)$ satisfy NA for any $S \subseteq [m]$ and any $a = (a_i)_{i \in S}$. For $S = \emptyset$, CNA is equivalent to NA. It was shown by Dubhashi et al.~\cite{DubhashiJR07} that if one rounds the variables according to a predefined linear order over the variables $\succ$ (i.e., if one always chooses for rounding the two fractional variables which are earliest in $\succ$), then the resulting distribution satisfies CNA. Yet, the requirement of following a predefined linear order of variables is too restrictive for our needs. Then, Kramer et al.~\cite{KramerCR11} showed that DR following a predefined order on pairs of variables that implements a tournament tree returns a distribution satisfying CNA.

\end{description}  

In our analysis we will use a simpler version of the NA property, which nevertheless is expressive enough for our needs. We introduce the following property.

\begin{description}
\item[Binary Negative Association (BNA).] For each $S, Q \subseteq [m]$ with $S \cap Q = \emptyset$, $s = |S|$, and $q = |Q|$, and each two nondecreasing functions, $f\colon \{0, 1\}^s \to \{0, 1\}$ and $g\colon \{0, 1\}^q \to \{0, 1\}$, we have:
\begin{align*}
&\cov\big[f(Y_i\colon i \in S), g(Y_i\colon i \in Q)\big] \leq 0.
\end{align*}
\end{description}

From the definitions it is easy to see that CNA $\implies$ NA $\implies$ BNA.

\subsection{BNA is Strictly Stronger than NC}

We now argue that BNA is a strictly stronger property than NC. First we show a straightforward inductive argument that BNA implies NC. Next we provide an example of a distribution that satisfies NC but not BNA. In fact, this distribution is generated by a not-careful-enough implementation of DR.

\begin{lemma}\label{lem:cov_for_binary}
For two binary random variables $X$, and $Y$, $X, Y \in \{0, 1\}$, the condition $\cov[X, Y] \leq 0$ is equivalent to $\prob{X = 1 \wedge Y = 1} \leq \prob{X = 1} \cdot \prob{Y = 1}$.
\end{lemma}
\begin{proof}
Observe that for binary variables, $X$ and $Y$, it holds that $\expected{X} = \prob{X = 1}$, $\expected{Y} = \prob{Y = 1}$, and $\expected{XY} = \prob{X = 1 \wedge Y = 1}$.
\end{proof}

\begin{lemma}\label{lem:bna_implies_nc}
Binary Negative Association of $(Y_i)_{i \in [m]}$ implies their Negative Correlation.
\end{lemma}
\begin{proof}
We will prove the NC property by induction on $|S|$. Clearly, the property holds for $|S|=1$. For an inductive step, we define two non-decreasing functions $f(Y_{i}\colon i \in S/\{j\}) = \bigwedge_{i \in S/\{j\}} (Y_{i} = 1)$ and $ g(Y_j) = (Y_j=1)$ for any $j \in S$.

\begin{eqnarray*}
\probBigPar{\bigwedge_{i \in S} (Y_i = 1)} &=& \probBigPar{\bigwedge_{i \in S/\{j\}} (Y_i = 1) \;\wedge\; Y_j=1} \\
&\stackrel{\text{BNA, Lemma~\ref{lem:cov_for_binary}}}{\leq}& \probBigPar{\bigwedge_{i \in S/\{j\}} (Y_i = 1)} \cdot \probBigPar{Y_j=1}\\
&\stackrel{\text{inductive assum.}}{\leq}& \prod_{i \in S}\prob{Y_i = 1}.
\end{eqnarray*} 

In order to bound the probability of $\bigwedge_{i \in S} (Y_i = 0)$ we define two other non-decreasing functions $f(Y_{i}\colon i \in S/\{j\}) = \bigvee_{i \in S/\{j\}} (Y_{i} > 0)$ and $ g(Y_j) = (Y_j > 0)$ for any $j \in S$.
\begin{eqnarray*}
&\probBigPar{\bigwedge_{i \in S} (Y_i = 0)}&  = \hspace{0.2cm}1 - \probBigPar{\bigvee_{i \in S} (Y_i > 0)} \\ 
&=& \hspace{-1.3cm} 1 - \left(\probBigPar{\bigvee_{i \in S/\{j\}} (Y_i > 0)} + \probBigPar{Y_j > 0} - \probBigPar{\bigvee_{i \in S/\{j\}} (Y_i > 0) \;\wedge\; Y_j > 0}\right) \\
&=& \hspace{-1.3cm} \probBigPar{\bigwedge_{i \in S/\{j\}} (Y_i = 0)} - \probBigPar{Y_j > 0} + \probBigPar{\bigvee_{i \in S/\{j\}} (Y_i > 0) \;\wedge\; Y_j > 0} \\
&\stackrel{\text{BNA, Lemma~\ref{lem:cov_for_binary}}}{\leq}&\hspace{-0.5cm} \probBigPar{\bigwedge_{i \in S/\{j\}} (Y_i = 0)} - \probBigPar{Y_j > 0} + \probBigPar{\bigvee_{i \in S/\{j\}} (Y_i > 0)} \cdot  \probBigPar{Y_j > 0} \\
&=& \hspace{-1.3cm} \probBigPar{\bigwedge_{i \in S/\{j\}} (Y_i = 0)} - \probBigPar{Y_j > 0} \cdot \probBigPar{\bigwedge_{i \in S/\{j\}} (Y_i = 0)} \\
&=& \hspace{-1.3cm} \probBigPar{\bigwedge_{i \in S/\{j\}} (Y_i = 0)} \cdot \probBigPar{Y_j=0} \stackrel{\text{inductive assum.}}{\leq} \prod_{i \in S}\prob{Y_i = 0}.
\end{eqnarray*}
\end{proof}

Note that the general formulation of DR does not specify how the pairs of fractional variables are selected. The proof in~\cite{Srinivasan01} that DR satisfies NC is independent of the method in which these pairs of fractional variables are selected. We will now show that, if these pairs are selected by an adaptive adversary who may take into account the way in which the previous pairs were rounded, then the BNA property may not hold (so, also neither NA nor CNA). Consider the following example.

\begin{example}\label{ex:no_bna}
Consider $m = 8, k = 4$, and the vector of variables $(y_{i})_{i \in [8]}$, all with the same initial value $\nicefrac{1}{2}$. Let $S = \{2, 3, 4\}$, $Q = \{5\}$, and:
\begin{align*}
f(Y_2, Y_3, Y_4) = 
\begin{cases}
1 & \quad \text{if~~} Y_2 + Y_3 + Y_4 \geq 2 \\
0 & \quad \text{otherwise} \\
\end{cases}
\quad
\quad
g(Y_5) = Y_5  \text{.}
\end{align*} 
Let $\alpha$ and $\beta$ denote the events that $Y_2 + Y_3 + Y_4 \geq 2$ and that $Y_5 = 1$, respectively. BNA would require that $\prob{\alpha \wedge \beta} \leq \prob{\alpha} \cdot \prob{\beta}$. Consider DR procedure as depicted in the following diagram (the paired variables are enclosed in rounded rectangles). First, we pair variables $y_1$ with $y_5$ and $y_2$ with $y_6$. The way in which the remaining variables are paired depends on the result of rounding within pairs $(y_1, y_5)$ and $(y_2, y_6)$. If $y_1$ and $y_2$ are both rounded to the same integer, then we pair $y_3$ with $y_7$ and $y_4$ with $y_8$. Otherwise, we pair $y_3$ with $y_4$ and $y_7$ with $y_8$.

\begin{figure}[t]
  \includegraphics[trim=0 0 0 0, clip, scale=1.0]{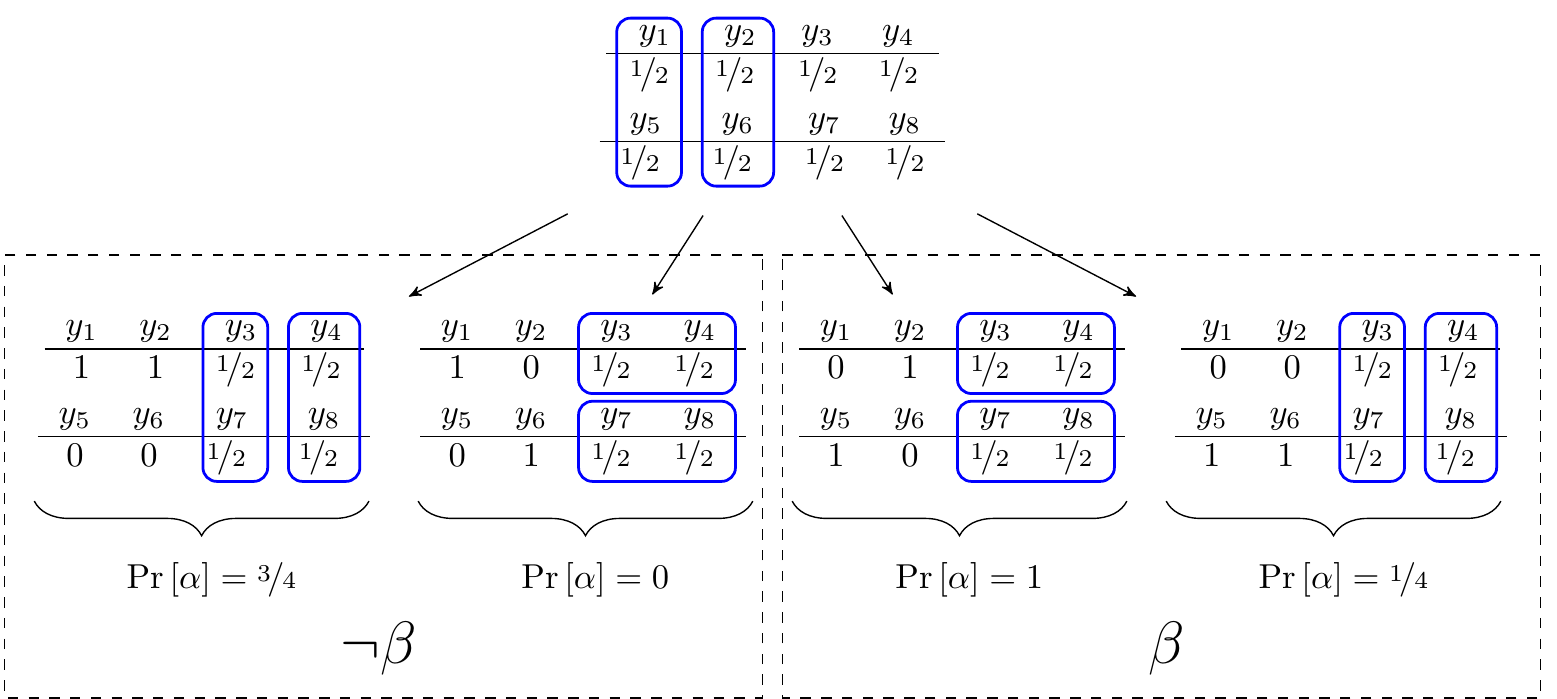}
  \caption{An illustration of Example~\ref{ex:no_bna}.}\label{fig:example_no_bna}
\end{figure}
Note that according to DR each rounding decision is taken with the same probability (e.g., when we pair variables $y_1$ with $y_5$, then the probabilities of $y_1$ and $y_5$ rounded to one is the same). Thus, we observe that $\prob{\alpha} = \nicefrac{1}{2}$, $\prob{\beta} = \nicefrac{1}{2}$, but $\prob{\alpha \wedge \beta} = \nicefrac{1}{4} + \nicefrac{1}{16}$.

\end{example}

Example~\ref{ex:no_bna} is simpler than the one given by Kramer et al.~\cite{KramerCR11}. Both examples show that NA is a strictly stronger property than NC. Kramer et al. use the set of 7 variables with initial values equal to $\nicefrac{3}{7}$, $k=4$, and a predefined order on pairs of variables. Our example uses an adaptive adversary who decides which pair of variables should be rounded in each step of the rounding procedure. Our example cannot be implemented by fixing an order on pairs of variables (hence it also cannot be implemented by fixing a tournament tree). Our 8 variables have marginal probabilities equal to $\nicefrac{1}{2}$, thus the example can be easily understood, and one does not need to calculate probabilities of choosing all ${7 \choose 4}$ 4-element sets. 

\subsection{Fixed Tournament Pairings Ensure BNA}

The method in which fractional elements are paired together can be thought of as a subset of rules of a sports tournament, in which losers drop out of the game, but winners remain and are being paired up for the following games.
The above example shows that an awkward adaptive pairing of remaining players may influence the value of certain functions on the subsets of players. We will show that if the competition is organized by a standard fixed upfront tournament tree,
then such manipulations are not possible, which allows to prove BNA for the outcome of the DR process following such tree.

Intuitively, the way in which the variables are paired should be, in some sense, independent of the result of previous roundings. We consider a fixed binary tree with $m$ leaves---each leaf containing one variable $y_{i}$ with value $y^*_{i}$, so that each variable is put in exactly one leaf; the other nodes are temporarily empty. 
In each step, the algorithm selects two nonempty nodes, say $n_1$ and $n_2$, with a common empty parent, and applies the basic step of the DR procedure to the two variables in nodes $n_1$ and $n_2$. As a result at least one of the variables becomes an integer. If one of the variables is still fractional, we promote this variable with its new value to the parent node. If both variables become integers (which happens when their sum is equal to one), we promote a fake variable $\bot$ to the parent node. When we compare any variable $v$ with $\bot$, we always promote $v$ with its current value to the parent node. An example run of such implementation of the DR procedure is depicted in Figure~\ref{fig:tournament_tree}.      

\begin{figure}
\centering
\begin{tikzpicture}[<-,>=stealth',
                   level 1/.style={sibling distance = 6cm/1,level distance = 1.5cm},
                   level 2/.style={sibling distance = 7.5cm/2,level distance = 1.5cm},
                   level 3/.style={sibling distance = 6cm/3,level distance = 1.5cm},
                   treenode/.style = {align=center, inner sep=2pt, text centered,
                                      font=\sffamily, black, text width=4.0em, text height=0.9em}] 

\node [align=center, inner sep=0pt, text centered, font=\sffamily] {}
    node [treenode, sibling distance = 5cm] {$\bot$}
    child{node [treenode] {$y_1 = 0.2$} 
            child{ node [treenode] {$y_1 = 0.2$} 
            	child{ node [treenode] {$y_1 = 0.7$}}
		child{ node [treenode] {$y_2 = 0.5$}}
            }
            child{ node [treenode] {$\bot$}
            	child{ node [treenode] {$y_3 = 0.5$}}
		child{ node [treenode] {$y_4 = 0.5$}}
            }                            
    }
    child{ node [treenode] {$y_6 = 0.8$}
				child{ node [treenode] {$y_5 = 0.2$}}
				child{ node [treenode] {$y_6 = 0.6$}}
   }
; 
\end{tikzpicture}
\caption{An example run of DR using a tournament tree structure. In this example the result is: $y_2, y_3, y_6 = 1$, and $y_1, y_4, y_5 = 0$.}\label{fig:tournament_tree}
\end{figure}
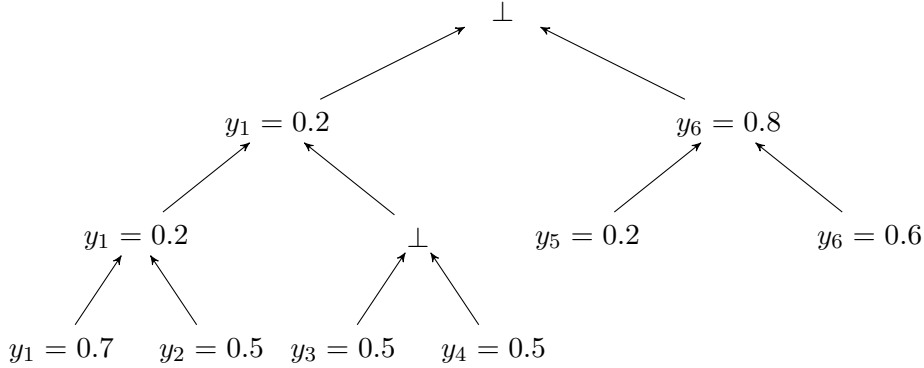

Hereinafter we assume that the DR procedure uses a fixed tournament tree structure, as described above.

\begin{theorem}\label{thm:dep_rounding_bna}
The DR algorithm using a tournament tree structure guarantees BNA.
\end{theorem}

The proof follows from Theorem 5 in~\cite{KramerCR11}. The theorem says that DR using a tournament tree structure produces distributions satisfying the NA property; clearly, NA implies BNA. However, to make the paper self-contained, we provide our inductive proof of Theorem~\ref{thm:dep_rounding_bna} in the remainder of the section. Our proof uses induction on the number of fractional variables; Kramer et al.~\cite{KramerCR11} use induction on the number of leaves in the tournament tree. While the two proofs use similar ideas and are of a similar difficulty, our proof is slightly more direct and shorter.

\begin{proof}[Proof of Theorem~\ref{thm:dep_rounding_bna}]
Recall that $Y_i$ is a random variable that indicates whether or not the described DR procedure rounds $y_i$ to 1. Let $S, Q \subseteq [m]$ with $S \cap Q = \emptyset$, $s = |S|$, and $q = |Q|$, and let $f\colon \{0, 1\}^s \to \{0, 1\}$ and $g\colon \{0, 1\}^q \to \{0, 1\}$ be two nondecreasing functions. Let $\alpha$ and $\beta$ denote the events that $f(Y_i\colon i \in S) = 1$ and $g(Y_i\colon i \in Q) = 1$, respectively.

For a vector $\overline{y}$ of $m$ values, which represents the values of the variables $(y_i)_{i \in [m]}$ that appear during our rounding procedure by $\probBigPar{E | \overline{y}}$ we denote the probability that an event $E$ occurs under the condition that we have reached the point of the rounding algorithm where the variables $(y_i)_{i \in [m]}$ have values indicated by $\overline{y}$. By Lemma~\ref{lem:cov_for_binary} it is sufficient to show that the following inequality holds for each $\overline{y}$:
\begin{align}\label{eq:_cond_negative_association}
\probBigPar{\alpha \wedge \beta | \overline{y}} \leq \probBigPar{\alpha | \overline{y}} \cdot \probBigPar{\beta | \overline{y}}.
\end{align}
We will prove this statement by induction on the number of fractional variables in $\overline{y}$. If $\overline{y}$ contains only integer variables, then it is clear that Inequality~\eqref{eq:_cond_negative_association} is satisfied. Now assume that Inequality~\eqref{eq:_cond_negative_association} is satisfied whenever $\overline{y}$ contains at most $\ell$ fractional values. We will show that Inequality~\eqref{eq:_cond_negative_association} is also satisfied when $\overline{y}$ contains $\ell+1$ fractional values. Let $\overline{y}$ be such vector. Consider a single step of our algorithm, where the two variables $y_i$ and $y_j$ are paired. Let $E_i$ and $E_j$ denote the events that, respectively, $y_i$ and $y_j$, is increased. Similarly, let $\overline{y(i)}$ and $\overline{y(j)}$ denote the vectors of the values of the variables $(y_i)_{i \in [m]}$ when, respectively, $y_i$ and $y_j$ is increased. We have:
\begin{align*}
\probBigPar{\alpha \wedge \beta | \overline{y}} = \probBigPar{\alpha \wedge \beta | \overline{y(i)}} \cdot \probBigPar{E_i} + \probBigPar{\alpha \wedge \beta | \overline{y(j)}} \cdot \probBigPar{E_j} \text{.}
\end{align*}
By our inductive assumption, it holds that:
\begin{align}\label{eq:_cond_negative_association2}
\probBigPar{\alpha \wedge \beta | \overline{y}} \leq \probBigPar{\alpha | \overline{y(i)}} \cdot \probBigPar{\beta | \overline{y(i)}} \cdot \probBigPar{E_i}
                                           + \probBigPar{\alpha | \overline{y(j)}} \cdot \probBigPar{\beta | \overline{y(j)}} \cdot \probBigPar{E_j} \text{.}
\end{align}
Now, we consider the following cases:
\begin{description}
\item[Case 1: $i, j \notin S$.] Observe that either the fake variable $\bot$ is promoted to the parent node or one of the variables: $y_i$ and $y_j$. Observe that irrespectively of which of the two variables is promoted to the parent, the promoted variable will always hold the same new value. Further, observe that the subsequent rounding steps do not depend on which variable has been promoted to the parent node, but only on the value of the promoted variable. Thus, the rounding within the pair of variables, $y_i$ and $y_j$, affects only the probability of events including $Y_i$ or $Y_j$ (here we use the assumption that the tournament tree is fixed; the way in which the variables are paired does not depend on the result of rounding within the pair $(y_i, y_j)$). In particular, $\probBigPar{\alpha | \overline{y(i)}} = \probBigPar{\alpha | \overline{y}}$ and  $\probBigPar{\alpha | \overline{y(j)}} = \probBigPar{\alpha | \overline{y}}$. We can rewrite Inequality~\eqref{eq:_cond_negative_association2} as follows:
\begin{align*}
\probBigPar{\alpha \wedge \beta | \overline{y}} &\leq \probBigPar{\alpha | \overline{y}} \cdot \probBigPar{\beta | \overline{y(i)}} \cdot \probBigPar{E_i}
                                           + \probBigPar{\alpha | \overline{y}} \cdot \probBigPar{\beta | \overline{y(j)}} \cdot \probBigPar{E_j} \\
                                          &= \probBigPar{\alpha | \overline{y}}\Big( \probBigPar{\beta | \overline{y(i)}} \cdot \probBigPar{E_i} + \probBigPar{\beta | \overline{y(j)}} \cdot \probBigPar{E_j}\Big) \\
                                          &= \probBigPar{\alpha | \overline{y}} \cdot \probBigPar{\beta | \overline{y}} \text{.}
\end{align*}

\item[Case 2: $i, j \notin Q$.] The same reasoning but applied to $\beta$ rather than to $\alpha$, leads to the same conclusion.

\item[Case 3: $i \in S$ and $j \in Q$ (the case when $i \in Q$ and $j \in S$ is symmetric).] As a result of rounding, one of the variables, $y_i$ and $y_j$, increases, and the other one decreases. Let us analyze what happens when $y_i$ increases and $y_j$ decreases, i.e., when event $E_i$ occurs. By the same reasoning as in Case 1, we infer that the fact that $y_i$ is increased does not influence the further process of rounding other variables than $y_i$ and $y_j$. At the same time when $y_i$ is increased, it becomes more likely that this variable will eventually become one, in comparison to the case when $y_j$ is increased:
\begin{enumerate}[(i)]
\item If $y_i + y_j \geq 1$, then $y_i$ being increased means that $y_i$ becomes one right away.
\item Otherwise, i.e., if $y_i + y_j < 1$: if $y_i$ is increased, it is still positive so it is still possible that it will eventually become one. On the other hand, if $y_j$ is increased, then $y_i$ is rounded down to zero, which makes it impossible for $y_i$ to become one.
\end{enumerate}
Since the function $f$ is nondecreasing we infer that $\probBigPar{\alpha | \overline{y(i)}} \geq \probBigPar{\alpha | \overline{y(j)}}$. 
The same reasoning allows us to conclude that $\probBigPar{\beta | \overline{y(i)}} \leq \probBigPar{\beta | \overline{y(j)}}$. This is summarized in the following claim:
\begin{claim}\label{claim:probabilities_estimations1}
$\probBigPar{\alpha | \overline{y(i)}} \geq \probBigPar{\alpha | \overline{y(j)}}$ and $\probBigPar{\beta | \overline{y(i)}} \leq \probBigPar{\beta | \overline{y(j)}}$.
\end{claim}
At the same time:
\begin{align}\label{eq:sum_of_cond_prob}
\begin{split}
\probBigPar{\beta | \overline{y(i)}} \cdot \probBigPar{E_i} + \probBigPar{\beta | \overline{y(j)}} \cdot \probBigPar{E_j} &= \probBigPar{\beta | \overline{y}} \\
   &= \probBigPar{\beta | \overline{y}} \big(\probBigPar{E_i} + \probBigPar{E_j}\big).
\end{split}
\end{align}
\begin{claim}\label{claim:probabilities_estimations2}
It holds that:
\begin{enumerate}[(i)]
\item $\probBigPar{\beta | \overline{y(i)}}\probBigPar{E_i} \leq \probBigPar{\beta | \overline{y}} \probBigPar{E_i}$, and
\item $\probBigPar{\beta | \overline{y(j)}}\probBigPar{E_j} \geq \probBigPar{\beta | \overline{y}} \probBigPar{E_j}$.
\end{enumerate}
\end{claim}
\begin{proof}[Proof of Claim~\ref{claim:probabilities_estimations2}]
For the sake of contradiction, let us assume that one of these inequalities is not satisfied, say assume that $\probBigPar{\beta | \overline{y(i)}}\probBigPar{E_i} > \probBigPar{\beta | \overline{y}} \probBigPar{E_i}$. By Equality~\eqref{eq:sum_of_cond_prob} we get that also $\probBigPar{\beta | \overline{y(j)}}\probBigPar{E_j} < \probBigPar{\beta | \overline{y}} \probBigPar{E_j}$. In these two conditions we can reduce the factors $\probBigPar{E_i}$ and $\probBigPar{E_j}$, respectively, and obtain that $\probBigPar{\beta | \overline{y(i)}} > \probBigPar{\beta | \overline{y}}$ and $\probBigPar{\beta | \overline{y(j)}} < \probBigPar{\beta | \overline{y}}$. By combining these two inequalities, we get that $\probBigPar{\beta | \overline{y(i)}} > \probBigPar{\beta | \overline{y(j)}}$, which contradicts Claim~\ref{claim:probabilities_estimations1}.
\end{proof}

Now, we continue the proof of Theorem~\ref{thm:dep_rounding_bna}. We will apply Lemma~\ref{lemma:simple_inequality} with:
\begin{enumerate}[(i)]
\item $a_1 = \probBigPar{\alpha | \overline{y(i)}}$, $a_2 = \probBigPar{\alpha | \overline{y(j)}}$,
\item $b_1 = \probBigPar{\beta | \overline{y(i)}} \probBigPar{E_i}$, $b_2 = \probBigPar{\beta | \overline{y(j)}} \probBigPar{E_j}$,
\item $c_1 = \probBigPar{\beta | \overline{y}} \probBigPar{E_i}$, and $c_2 = \probBigPar{\beta | \overline{y}} \probBigPar{E_j}$.
\end{enumerate}
($a_1 \geq a_2$ by Claim~\ref{claim:probabilities_estimations1}; $c_1 \geq b_1$ and $b_2 \geq c_2$, by Claim~\ref{claim:probabilities_estimations2}; $b_1 + b_2 = c_1 + c_2$ by Equality~\eqref{eq:sum_of_cond_prob}).
We get that:
\begin{align*}
&\probBigPar{\alpha | \overline{y(i)}} \cdot \probBigPar{\beta | \overline{y(i)}} \cdot \probBigPar{E_i} + \probBigPar{\alpha | \overline{y(j)}} \cdot \probBigPar{\beta | \overline{y(j)}} \cdot \probBigPar{E_j} \leq \\
&\probBigPar{\alpha | \overline{y(i)}} \cdot \probBigPar{\beta | \overline{y}} \cdot \probBigPar{E_i} + \probBigPar{\alpha | \overline{y(j)}} \cdot \probBigPar{\beta | \overline{y}} \cdot \probBigPar{E_j}.
\end{align*}

Combining the above inequality with Inequality~\eqref{eq:_cond_negative_association2}, we infer that:
\begin{align*}
\probBigPar{\alpha \wedge \beta | \overline{y}} &\leq \probBigPar{\alpha | \overline{y(i)}} \cdot \probBigPar{\beta | \overline{y}} \cdot \probBigPar{E_i}
                                           + \probBigPar{\alpha | \overline{y(j)}} \cdot \probBigPar{\beta | \overline{y}} \cdot \probBigPar{E_j} \\
                                &= \probBigPar{\beta | \overline{y}} \Big( \probBigPar{\alpha | \overline{y(i)}} \cdot \probBigPar{E_i} + \probBigPar{\alpha | \overline{y(j)}} \cdot \probBigPar{E_j} \Big) \\
                                &= \probBigPar{\beta | \overline{y}} \cdot \probBigPar{\alpha | \overline{y}}
\end{align*}
\end{description}
This proves the inductive step and completes the proof.
\end{proof}

\section{Useful Lemmas}\label{sec:useful_lemmas}

\begin{theorem}[Theorem 1.16 from~\cite{auger2011theory}]\label{thm:chernoff}
Let $X_1, X_2, \dots, X_n$ be negatively correlated binary random variables. Let $X = \sum_{i=1}^{n} X_i$. Then $X$ satisfies the Chernoff-Hoeffding bounds for $\delta \in [0,1]$:
\[ \probBigPar{X \leq (1-\delta)\expected{X}} \leq \left( \frac{e^{-\delta}}{(1-\delta)^{(1-\delta)}} \right)^\mu .\]
\end{theorem}

\begin{lemma}\label{lem:sai_over_sbi_leq_max_ai_over_bi}
For any sequence $(a_i)_{i \in [n]}$ and $(b_i)_{i \in [n]}$, $b_i > 0$, it holds:
\[\frac{\sum_{i=1}^{n} a_i}{\sum_{i=1}^{n} b_i} \leq \max_{i \in \{1,2,\dots,n\}}\frac{a_i}{b_i}.\]
\end{lemma}
\begin{proof}
\begin{eqnarray*}
\frac{\sum_{i=1}^{n} a_i}{\sum_{i=1}^{n} b_i} &=& \sum_{j=1}^{n}\frac{a_j}{\sum_{i=1}^{n} b_i} = \sum_{j=1}^{n}\frac{a_j}{b_j} \cdot \frac{b_j}{\sum_{i=1}^{n} b_i} \leq \sum_{j=1}^{n} \left(\max_{i \in \{1,2,\dots,n\}}\frac{a_i}{b_i}\right) \frac{b_j}{\sum_{i=1}^{n} b_i} =\\
&=& \left(\max_{i \in \{1,2,\dots,n\}}\frac{a_i}{b_i}\right) \sum_{j=1}^{n} \frac{b_j}{\sum_{i=1}^{n} b_i} = \max_{i \in \{1,2,\dots,n\}}\frac{a_i}{b_i}.
\end{eqnarray*}
\end{proof}

\begin{lemma}\label{lem:ca_over_cb_leq_ab}
For any non-decreasing sequence $(c_i)_{i \in \{1,2,\dots,n\}}$, $c_i >0$ and any non-increasing sequence $(a_i)_{i \in \{1,2,\dots,n\}}$ it holds:
\[\frac{\sum_{i=1}^{n} a_i c_i}{\sum_{i=1}^{n} c_i} \leq \frac{1}{n} \sum_{i=1}^{n} a_i.\]
\end{lemma}
\begin{proof}
We prove that by induction. Clearly, we have equality for $n=1$. We assume that
\begin{equation*}
  \frac{\sum_{i=1}^{n-1} a_i c_i}{\sum_{i=1}^{n-1} c_i} \leq \frac{1}{n-1} \sum_{i=1}^{n-1} a_i.
\end{equation*}
It is equivalent to
\begin{equation}\label{ineq:lem_seq_inductive_assum}
  (n-1) \cdot \sum_{i=1}^{n-1} a_i c_i \leq \left(\sum_{i=1}^{n-1} a_i\right)\cdot\left(\sum_{i=1}^{n-1} c_i\right).
\end{equation}
We would like to show that
\[ n \cdot \sum_{i=1}^{n} a_i c_i \leq \left(\sum_{i=1}^{n} a_i\right)\cdot\left(\sum_{i=1}^{n} c_i\right). \]
We have the following equivalent inequalities:
\begin{eqnarray*}
  0 &\leq& \left(\sum_{i=1}^{n-1} a_i\right)\cdot\left(\sum_{i=1}^{n-1} c_i\right) + a_n \cdot \sum_{i=1}^{n-1} c_i + c_n \cdot \sum_{i=1}^{n-1} a_i + a_n \cdot c_n - n \cdot \sum_{i=1}^{n-1} a_i c_i - n \cdot a_n \cdot c_n, \\
  0 &\leq& \left[ \left( \sum_{i=1}^{n-1} a_i \right) \cdot \left( \sum_{i=1}^{n-1} c_i \right) - (n-1) \cdot \sum_{i=1}^{n-1} a_i c_i \right] + \sum_{i=1}^{n-1} \left(a_n \cdot c_i + c_n \cdot a_i - a_n \cdot c_n - a_i \cdot c_i \right), \\
  0 &\leq& \left[ \left( \sum_{i=1}^{n-1} a_i \right) \cdot \left( \sum_{i=1}^{n-1} c_i \right) - (n-1) \cdot \sum_{i=1}^{n-1} a_i c_i \right] + \sum_{i=1}^{n-1} \left( a_i - a_n \right)\left( c_n - c_i \right). \\
\end{eqnarray*}
Using the inductive assumption~\eqref{ineq:lem_seq_inductive_assum} and monotonicity of sequences, i.e., $0 \leq a_i - a_n$, $0 \leq c_n - c_i$ we finish the proof.
\end{proof}

\begin{lemma}\label{lemma:simple_inequality}
Let $a_1, a_2, b_1, b_2, c_1, c_2 \in \reals$ be such that $a_1 \geq a_2$, $c_1 \geq b_1$, $b_2 \geq c_2$, and $b_1 + b_2 = c_1 + c_2$. It holds that:
$a_1 c_1 + a_2 c_2 \geq a_1 b_1 + a_2 b_2$.
\end{lemma}
\begin{proof}
We have that $b_2 - c_2 = c_1 - b_1 \geq 0$, and so $a_2(b_2 - c_2) \leq a_1(c_1 - b_1)$, which can be reformulated as $a_1(c_1 - b_1) + a_2(c_2 - b_2) \geq 0$. Thus:
\begin{align*}
a_1 c_1 + a_2 c_2 = a_1 b_1 + a_2 b_2 + a_1 (c_1 - b_1) + a_2 (c_2 - b_2) \geq a_1 b_1 + a_2 b_2 \text{.}
\end{align*}
\end{proof}

\section{Omitted Proofs from Section~\ref{sec:main_approximation}}\label{appen:omitted_:main_approximation}

\begin{lemma}\label{lem:bin_neg_asoc_Zr}
Distribution of $\{Z_1,Z_2, \dots, Z_{\nicefrac{k}{\epsilon}} \}$ satisfies Binary Negative Association.
\end{lemma}
\begin{proof}[Proof sketch]
Note that DR procedure on $(Y_i)_{i \in [m]}$ and then independent choice of $(Z_r)_{r \in \sub(i)}$ for each $i \in [m]$ is equivalent to the following implementation of DR on $(Z_r)_{r \in \{1,2,\dots,\nicefrac{k}{\epsilon}\}}$. First, for each $i \in [m]$ $(Z_r)_{r \in \sub(i)}$ are processed until obtaining a single non-zero variable that is equivalent to $y_i$. Then, in the second phase the rounding proceeds as if it had started from the $y_i$ variables. 
Since this process altogether is an implementation of a single DR procedure with fixed tournament tree starting from $(Z_r)_{r \in \{1,2,\dots,\nicefrac{k}{\epsilon}\}}$ variables, we can simply apply Theorem~\ref{thm:dep_rounding_bna} and get the statement of the lemma.

At this point we note that the result of Dubhashi et al.~\cite{DubhashiJR07} is not sufficient for proving our lemma. They have proved that DR following a predefined order of variables (which can be viewed as a linear tournament tree) returns distributions satisfying the CNA property. Here, however, we need to have at least a "two-stage" linear tournament: the first linear tournament on variables $(Z_r)_{r \in \sub(i)}$ and the second tournament on winning variables from the first tournament.
\end{proof}

\begin{proof}[Proof of Lemma~\ref{lem:Hr_chernoff}]
Let us fix $\ell \in [k]$, $t \in [\ell-1]$ and $r \in \{\nicefrac{(\ell-1)}{\epsilon}+1,\nicefrac{(\ell-1)}{\epsilon}+2,\dots,\nicefrac{\ell}{\epsilon}\}$. We have
\begin{eqnarray}
H_r(t-1) &\stackrel{\eqref{eq:Hr_def}}{=}& \probBigPar{\Sum_{r'=1}^{r-1} Z_{r'} \leq t-1 \;\bigg|\; Z_{r} = 1} = \probBigPar{\Sum_{r'=r}^{\nicefrac{k}{\epsilon}} Z_{r'} \geq k-(t-1) \;\bigg|\; Z_{r} = 1} = \nonumber\\
&=& \probBigPar{\Sum_{r'=r+1}^{\nicefrac{k}{\epsilon}} Z_{r'} \geq k-t \;\bigg|\; Z_{r} = 1}.\label{ineq:Fz1ub}
\end{eqnarray}
We now exploit Binary Negative Association of variables $Z_i$ (Lemma~\ref{lem:bin_neg_asoc_Zr}). By setting $S = \{r+1,r+2, \dots, \nicefrac{k}{\epsilon}\}, Q = \{r\}, f(a_1,a_2, \dots a_s) = \mathbbm{1}\Big\{\sum_{i=1}^{|S|} a_i \geq k-t\Big\}$ and $g(a) = a$ we obtain:

\[ 0 \geq \cov\big[f(Z_{r'}\colon r' \in S), g(Z_{r'}\colon r' \in Q) \big] = \cov\left[\mathbbm{1}\left\{\sum_{r'=r+1}^{\nicefrac{k}{\epsilon}} Z_{r'} \geq k-t \right\}, Z_{r} \right]. \]
Since $f,g$ are binary and non-decreasing we can use Lemma~\ref{lem:cov_for_binary} to obtain an equivalent inequality:
\begin{equation}\label{ineq:pr_con_mult_pr} 
\probBigPar{\sum_{r'=r+1}^{\nicefrac{k}{\epsilon}} Z_{r'} \geq k-t \quad \wedge \quad Z_{r} = 1} \leq \probBigPar{\sum_{r'=r+1}^{\nicefrac{k}{\epsilon}} Z_{r'} \geq k-t } \cdot \probBigPar{Z_{r} = 1}.
\end{equation}
Therefore,
\begin{eqnarray}
H_r(t-1) &\stackrel{\eqref{ineq:Fz1ub}}{\leq}& \probBigPar{\Sum_{r'=r+1}^{\nicefrac{k}{\epsilon}} Z_{r'} \geq k-t \;\bigg|\; Z_{r} = 1}  \nonumber\\
&=& \frac{\probBigPar{\Sum_{r'=r+1}^{\nicefrac{k}{\epsilon}} Z_{r'} \geq k-t \quad \wedge \quad Z_{r} = 1}}{\probBigPar{Z_{r} = 1}}
\nonumber\\
&\stackrel{\eqref{ineq:pr_con_mult_pr}}{\leq}& \probBigPar{\Sum_{r'=r+1}^{\nicefrac{k}{\epsilon}} Z_{r'} \geq k-t} = \probBigPar{\Sum_{r'=1}^{r} Z_{r'} \leq t}. \label{ineq:Fzm1_remove_conditioning}
\end{eqnarray}
Using Lemma~\ref{lem:bin_neg_asoc_Zr} and Lemma~\ref{lem:bna_implies_nc} we know that $(Z_r)_{r \in \{1,2,\dots,\nicefrac{k}{\epsilon}\}}$ are negatively correlated. What is more, $t$ is smaller than the expected value of the sum
\[ t \leq \ell-1 = (\ell-1+\epsilon)-\epsilon \leq r \cdot \epsilon - \epsilon < r \cdot \epsilon \stackrel{\eqref{eq:pr_zr1}}{=} \expectedBigPar{\Sum_{r'=1}^{r} Z_{r'}}, \]
Therefore, we can use Chernoff-Hoeffding bounds as follows
\begin{eqnarray*} H_r(t-1) \hspace{-0.8cm} &\stackrel{\eqref{ineq:Fzm1_remove_conditioning}}{\leq}& \probBigPar{\Sum_{r'=1}^{r} Z_{r'} \leq t } \\
&=& \probBigPar{\Sum_{r'=1}^{r} Z_{r'} < r \cdot \epsilon \cdot \left(1-\left(1-\frac{t}{r \cdot \epsilon}\right)\right) } \\
&\stackrel{\text{Theorem \ref{thm:chernoff}}}{\leq}& \left(\frac{e^{\frac{t}{r \cdot \epsilon}-1}}{\left(\frac{t}{r \cdot \epsilon}\right)^{ \frac{t}{r \cdot \epsilon} }} \right)^{r \cdot \epsilon} = \frac{e^{t-r \cdot \epsilon} \cdot (r \cdot \epsilon)^{t}}{t^t} \\
&=& e^{-r \cdot \epsilon} \cdot \left(\frac{e \cdot r \cdot \epsilon}{t} \right)^t.
\end{eqnarray*}

\end{proof}

\begin{proof}[Proof of Lemma~\ref{lem:final_apx_ratio}]
\begin{eqnarray}
\frac{\expected{\cost_j(Y)}}{\OPTLP_j} \hspace{-0.3cm} &\stackrel{\eqref{ineq:approx_ratio},\eqref{eq:Er},\eqref{ineq:Er}}{\leq}& \hspace{0.3cm}\max_{\ell \in [k]} \; \epsilon \cdot \ell \cdot \hspace{-0.5cm} \Sum_{r=\nicefrac{(\ell-1)}{\epsilon} + 1}^{\nicefrac{\ell}{\epsilon}} \; \left( \Sum_{t=1}^{\ell-1} \left( \frac{1}{t(t+1)} \cdot H_r(t-1) \right) + \frac{1}{\ell} \right) \nonumber\\
&=& \hspace{-0.3cm} 1+ \max_{\ell \in [k]} \; \epsilon \cdot \ell \cdot \hspace{-0.5cm} \Sum_{r=\nicefrac{(\ell-1)}{\epsilon} + 1}^{\nicefrac{\ell}{\epsilon}} \left( \Sum_{t=1}^{\ell-1} \frac{1}{t(t+1)} \cdot H_r(t-1) \right) \nonumber\\
&\stackrel{\text{Lemma~\ref{lem:Hr_chernoff}}}{\leq}& \hspace{-0.3cm} 1 + \max_{\ell \in [k]} \; \epsilon \cdot \ell \cdot \hspace{-0.5cm} \Sum_{r = \nicefrac{(\ell-1)}{\epsilon}+1}^{\nicefrac{\ell}{\epsilon}} \; \left( \Sum_{t=1}^{\ell-1} \frac{1}{t(t+1)} \cdot e^{-r \cdot \epsilon} \cdot \left(\frac{e \cdot r \cdot \epsilon}{t} \right)^t  \right) \nonumber \\
&=& \hspace{-0.6cm} 1 + \max_{\ell \in [k]} \; \ell \cdot \Sum_{t=1}^{\ell-1} \frac{1}{t(t+1)} \cdot \frac{e^t}{t^t} \cdot \left( \Sum_{r = \nicefrac{(\ell-1)}{\epsilon}+1}^{\nicefrac{\ell}{\epsilon}} \; \epsilon \cdot e^{-r \cdot \epsilon} \cdot \left( r \cdot \epsilon \right)^{t} \right) \nonumber\\
&=& \hspace{-0.6cm} 1 + \max_{\ell \in [k]} \; \ell \cdot \Sum_{t=1}^{\ell-1} \frac{1}{t(t+1)} \cdot \frac{e^t}{t^t} \cdot \left( \Sum_{r=\nicefrac{(\ell-1)}{\epsilon} + 1}^{\nicefrac{\ell}{\epsilon} - 1} \; \int_{r \cdot \epsilon - \epsilon}^{r \cdot \epsilon} e^{-r \cdot \epsilon} \cdot \left( r \cdot \epsilon \right)^{t} dx\right) \leq \nonumber
\end{eqnarray}
we now use an upper bound on the most interior sum by an integral of the function $f_t(x) = e^{-x} \cdot x^t$. Note that $f'_t(x) = e^{-x} \cdot x^{t-1} \cdot (t-x) \leq 0$ for $1 \leq t \leq \ell-1 \leq x$, so the function $f$ is non-increasing. Therefore
\begin{eqnarray}
&\leq& 1 + \max_{\ell \in [k]} \; \ell \cdot \Sum_{t=1}^{\ell-1} \frac{1}{t(t+1)} \cdot \frac{e^t}{t^t} \cdot \left( \int_{\ell-1}^{\ell} e^{-x} \cdot x^t dx\right).\label{ineq:ratio_with_integral}
\end{eqnarray}
To bound the above expression we first numerically evaluate it for $\ell \in \{1,2,\dots,88\}$ and obtain
\[ 1 + \max_{\ell \in \{1,2,\dots,88\}} \; \ell \cdot \Sum_{t=1}^{\ell-1} \frac{1}{t(t+1)} \cdot \frac{e^t}{t^t} \cdot \left( \int_{\ell-1}^{\ell} e^{-x} \cdot x^t dx \right) < \ratiohkm. \]
It remains to bound the expression for $\ell \in \{89,90,\dots,k\}$, which we do by the following estimation:
\begin{eqnarray*}
1 + \ell \cdot \Sum_{t=1}^{\ell-1} \frac{1}{t(t+1)} \cdot \frac{e^t}{t^t} \cdot \left( \int_{\ell-1}^{\ell} e^{-x} \cdot x^t dx \right) \hspace{-1.3cm} &\leq& \hspace{-1cm} 1 + \ell \cdot \Sum_{t=1}^{\ell-1} \frac{1}{t(t+1)} \cdot \frac{e^t}{t^t} \cdot e^{-(\ell-1)} \cdot \ell^t \\
          &\stackrel{\text{Stirling}}{\leq}& \hspace{-1cm} 1 + \ell \cdot \Sum_{t=1}^{\ell-1} \frac{1}{t(t+1)} \cdot \frac{\sqrt{2 \pi t} \cdot e^{\frac{1}{12t}}}{t!} \cdot e^{-(\ell-1)} \cdot \ell^t \\
          &\leq& \hspace{-1cm} 1 + \sqrt{2 \pi} \cdot e^{\frac{1}{12}} \cdot e^{-(\ell-1)} \cdot \frac{1}{\sqrt{\ell}} \cdot \Sum_{t=1}^{\ell-1} \frac{\ell^{t+1}}{(t+1)!} \cdot \frac{\sqrt{\ell}}{\sqrt{t}} \\
          &\leq& \hspace{-1cm} 1 + \sqrt{2 \pi} \cdot e^{\frac{13}{12}} \cdot e^{-\ell} \cdot \frac{1}{\sqrt{\ell}} \cdot \Sum_{t=1}^{\ell-1} \frac{\ell^{t+1}}{(t+1)!} \cdot \frac{\ell}{t} \\
          &\leq& \hspace{-1cm} 1 + 3\sqrt{2 \pi} \cdot e^{\frac{13}{12}} \cdot e^{-\ell} \cdot \frac{1}{\sqrt{\ell}} \cdot \Sum_{t=1}^{\ell-1} \frac{\ell^{t+1}}{(t+1)!} \cdot \frac{\ell}{t+2} \\
          &\stackrel{\text{Taylor series for~}e^\ell}{\leq}& \hspace{-0.1cm} 1 + 3\sqrt{2 \pi} \cdot e^{\frac{13}{12}} \cdot e^{-\ell} \cdot \frac{1}{\sqrt{\ell}} \cdot e^\ell \\
          &=& \hspace{-1cm} 1 + 3\sqrt{2 \pi} \cdot e^{\frac{13}{12}} \cdot \frac{1}{\sqrt{\ell}} < 2.3551 < \ratiohkm.
\end{eqnarray*}

The maximum is obtained for $\ell=4$ (see Figure~\ref{fig:apx_ratio}).

\end{proof}

\begin{figure}[th]
  \hspace{-2cm}
  \begin{minipage}[c]{0.9\textwidth}
      \includegraphics[trim=0 20 0 0, clip, scale=0.6]{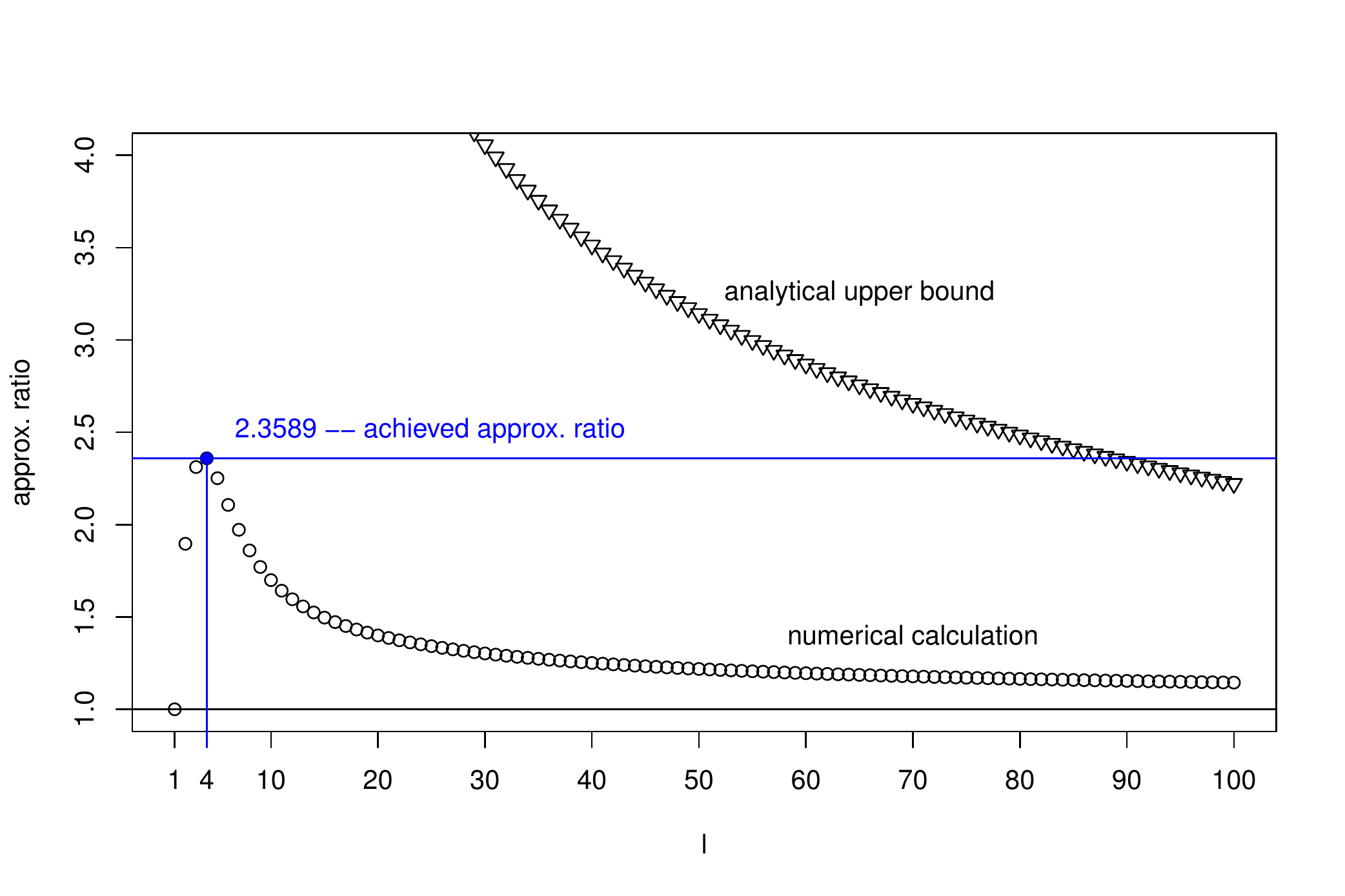}
  \end{minipage}
  \hspace{-2cm}
  \begin{minipage}[c]{0.05\textwidth}
  \vspace{0.9cm}
  \begin{tabular}{|l|l|}
    \hline
 {\bf $\ell$} & {\bf ratio}  \\ \hline
 1         & 1    \\ \hline
 2         & 1.90 \\ \hline
 3         & 2.32 \\ \hline
 {\bf\textcolor{blue}{4}}   & {\bf\textcolor{blue}{\ratiohkm}} \\ \hline
 5         & 2.26 \\ \hline
 6         & 2.11 \\ \hline
 7         & 1.98 \\ \hline
 8         & 1.86 \\ \hline
 9         & 1.78 \\ \hline
 10        & 1.70 \\ \hline
 11        & 1.65 \\ \hline
 12        & 1.60 \\ \hline
  \end{tabular}
\end{minipage}
\caption{The numerical and the analytical upper bound on the approximation ratio on intervals $(\ell-1,\ell)$, for each $\ell \in [k]$.}
\label{fig:apx_ratio}
\end{figure}

\section{Omitted Proofs from Section~\ref{sec:reduction_to_fault_tol}}
\label{sec:apxD}

\begin{proof}[Proof of Theorem~\ref{thm:93-apx-for-multi}]
We reduce an instance of \probftkmedmulti to an instance of \probftkmed by replacing the multiple $m_j$ of a client $D_j$ with $m_j$ clients in the same location and with the same connectivity requirement (we will call such clients clones of $D_j$). Observe that there exists an optimal solution in which each clone of the same client is connected to the same set of open facilities. Next, we run the $93$-approximation algorithm of Hajiaghayi~et~al.~\cite{HajiaghayiHLLS16} on such a constructed instance with clones. It is apparent that the solution that we obtain by following this procedure approximates the original instance with the ratio of~93. However, the issue is that $m_j$ can be exponential in the number of clients in the original instance, and so the most straightforward implementation of our reduction does not run in polynomial-time. To deal with that we will efficiently encode the reduced instance, and we will show that the algorithm of Hajiaghayi~et~al. can be adapted to run on such encoded instances. We proceed as follows.

First, we solve the LP part of the original algorithm~\cite{HajiaghayiHLLS16} with the additional multiplicative factors $\{m_j\}_{D_j \in \mathcal{D}}$ added to the objective function. From the solution to the LP, $(y_i)_{i \in [m]}$ with $\sum_{i=1}^m y_i = k$, we construct an optimal assignment of the clients to the facilities. We encode such an assignment efficiently by grouping all clones of the same client into a single cluster and storing an assignment for a single client for each cluster only (we call such a client the {\it representative} of the cluster). In particular, note that all clones in the same cluster have the same assignments and so, they all have the same average and maximal assignment costs. We use this property in the next step of the original algorithm: {\it creating bundles} of volume $1$~\cite[Algorithm~1]{HajiaghayiHLLS16}. By a careful analysis of this algorithm we can observe that no new bundles are created for a cloned client (lines 5 and 6 of~\cite[Algorithm~1]{HajiaghayiHLLS16}) and so that the cloned clients can be considered in bunches. 

Next, as in the original algorithm, we divide the clients into {\it safe} and {\it dangerous} by the criterion on the ratio of the maximal and the average cost in the assignment vector. Intuitively, if the maximum is much higher than the average then the client is marked as dangerous (for a formal definition see~\cite[Section 2.2]{HajiaghayiHLLS16}), otherwise it is considered safe. Hence, the clones of the same client are either all safe or all dangerous. In the latter case they are also {\it in conflict}: they are close and they have the same connectivity requirements (for a definition also see~\cite[Section 2.2]{HajiaghayiHLLS16}). Thus, in the {\it filtering phase}~\cite[Algorithm~2]{HajiaghayiHLLS16}  either all the dangerous clones of the same client are filtered out or exactly one of them survives; without loss of generality we can assume that the representative of the cluster survives. In fact, this is the main reason why we can quite easy adapt the algorithm. The next step, that is building a laminar family~\cite[Algorithm~3]{HajiaghayiHLLS16}, is independent on clients that were filtered out, and so it can be performed on our efficiently encoded instance. The safe clients are not used later on by the algorithm (they are only the side effect of creating bundles and later on they only appear in the algorithm's analysis).
Finally, the rounding process of the algorithm (\cite[Section 2.3]{HajiaghayiHLLS16}) depends on the set of constructed bundles and on the set of filtered dangerous clients (and the induced laminar family), and as we discussed it is possible to construct each of the two families with efficient encoding. This completes the proof.

\end{proof}

\begin{proof}[Proof of Lemma~\ref{lem:costftkmi-gen}]
Let $C$ be an $\alpha$-approximate solution to $I'$. By $\ftkmedmulti(C,j)$ we denote be the total cost of the clients $D_{j,1}, D_{j,2}, \dots, D_{j,k}$ constructed through reduction from Figure~\ref{fig:reduction}. 
Similarly, let $\owakmed(C,j)$ be the cost of the client $D_j$ for $C$ in $I$.
For each client $D_j$ we have:
\begin{align*}
 \ftkmedmulti(C,j) &=\sum_{r = 1}^k m_{j,r} \cdot \left( \sum_{i=1}^r c_{i}^\rightarrow(C,j) \right) = \sum_{r = 1}^k \sum_{i=1}^r m_{j,r} \cdot c_{i}^\rightarrow(C,j) \\
 & = \sum_{r = 1}^k \sum_{i=1}^r m_{j,r} \cdot c_{i}^\rightarrow(C,j) = \sum_{i = 1}^k \sum_{r=i}^k m_{j,r} \cdot c_{i}^\rightarrow(C,j) \\
 &= \sum_{i = 1}^k \left( c_{i}^\rightarrow(C,j) \cdot \sum_{r=i}^k m_{j,r} \right) \\
 &= \sum_{i = 1}^k c_{i}^\rightarrow(C,j) \cdot w_i \cdot Q = Q \cdot \owakmed(C,j).
\end{align*}
Let $C^*_I$ and  $C^*_{I'}$ be optimal solutions for $I$ and $I'$, respectively. By the same reasoning, we have that:
\begin{align*}
 \ftkmedmulti(C^*_{I},j) = Q \cdot\owakmed(C^*_{I},j).
\end{align*}
And, thus, that:
\begin{align*}
  &\sum_{D_j \in \mathcal{D}} \owakmed(C,j) = \sum_{D_j \in \mathcal{D}} \frac{1}{Q}\ftkmedmulti(C,j) \\ &\leq \alpha\frac{1}{Q}\sum_{D_j \in \mathcal{D}} \ftkmedmulti(C^{*}_{I'},j) \leq \alpha\frac{1}{Q} \sum_{D_j \in \mathcal{D}} \ftkmedmulti(C^{*}_{I},j) \\ &= \alpha \sum_{D_j \in \mathcal{D}} \owakmed(C^{*}_{I},j) \text{.}
\end{align*}
This completes the proof.
\end{proof}

\section{Hardness of approximation}\label{sec:hardness_general}
In the main text we have shown that the \probowakmedian problem with the harmonic sequence of weights admits very good approximations.
In this section we show that for many other natural sequences of weights, the considered problem is hard to approximate, unless we introduce additional assumptions, such as the assumptions that the costs satisfy the triangle inequality. Our hardness results hold already for 0/1 costs.

Let us start by considering the \probowakmedian problem for a certain specific class of weights. For each $k \in \naturals$, let $w^{(k)} = \big(w^{(k)}_1, \ldots, w^{(k)}_k \big)$ be a sequence of weights used in \probowakmedian.
Fix $\lambda \in (0, 1)$. We say that the clients care only about the $\lambda$-fraction of facilities if for each $k$ it holds that $w^{(k)}_i = 0$ whenever $i > \lambda k$. For instance, we say that the clients care only about 90\% of facilities if the cost of each client from a set $C$ does not depend on the $10\%$ of worst facilities in $C$.  

First, we prove a simple result which says that if there exists $\lambda$ such that the clients care only about the $\lambda$-fraction of facilities, and if
the costs of clients from facilities can be arbitrary, in particular if they cannot be represented as distances satisfying the triangle inequality, then the problem does not admit any approximation.  

\begin{theorem}\label{thm:cc_no_approximation}
Fix $\lambda \in (0, 1)$ and consider the problem \probowakmedian where clients care only about the $\lambda$-fraction of facilities.
For any positive computable function $\alpha$, there exists no polynomial-time $\alpha$-approximation algorithm for \probowakmedian, unless $\p = \np$.
\end{theorem}
\begin{proof}
Let us fix a function $\alpha$ and for the sake of contradiction, let us assume that there exists a polynomial time $\alpha$-approximation algorithm $\calA$ for \probowakmedian. We will show that $\calA$ can be used to find exact solutions to the exact set cover problem, \textsc{X3C}. This will stay in contradiction with the fact that \textsc{X3C} is $\np$-hard~\cite{GareyJ79}.

Let $I$ be an instance of \textsc{X3C}, where we are given a set of $3n$ elements $E = \{e_1, \ldots, e_{3n}\}$, and a collection $\calS$ of subsets of $E$ such that each set in $\calS$ contains exactly 3 elements from $E$. We ask if there exists a subcollection $C$ of $n$ subsets from $\calS$ such that each element from $E$ belongs to exactly one set from $C$.

From $I$ we construct an instance of \probowakmedian in the following way. First, we set the size of the committee $k = \left\lceil \frac{n}{1-\lambda} \right\rceil$. Let $p$ be the index of the last positive weight in the sequence $w^{(k)}$. Since clients care only about the $\lambda$-fraction of facilities, we know that $k - p > k - \lambda k \geq n$. Our set of facilities consists of three groups $\Fcal = \calS \cup H \cup H'$, i.e., we have facilities which correspond to subsets from $\calS$ and two groups of dummy facilities. We set $|H| = p-1$ and $|H'| = k-n-p+1$. Our set of clients consist of two groups $\Dcal = E \cup G$, where $G$ is the set of dummy clients with $|G| = |H'|$. Let us now describe the preferences of clients over facilities. For each client $j \in \Fcal$ and each dummy facility $i \in H$ we set $c_{i, j} = 0$. Further, for each non-dummy client $j \in E$ and a non-dummy facility $i \in \calS$ we set $c_{i, j} = 0$ if and only if $j \in i$. Finally, we match dummy clients from $G$ with dummy facilities from $H'$ so that each client is matched to exactly one facility and each facility to exactly one client, and set $c_{i, j} = 0$ whenever $i$ and $j$ are matched. For all remaining pairs $(i, j)$ we set $c_{i, j} = 1$.

Let $C^*$ be an optimal set of facilities. We will show that the total cost of clients from $C^*$ is 0 if and only if there exists an exact cover for our initial instance $I$.

$(\implies)$ Assume there exists an exact cover in $I$---let $C$ denote the collection of $n$ subsets covering all the elements. If we set $C^* = C \cup H \cup H'$, then each client has exactly $p$ facilities with distance $0$. Note that $|C^*|=n+p-1+k-n-p+1=k$. For the remaining $k-p$ facilities the weights are equal to zero. Thus, the total cost of clients from $C^*$ is equal to 0.

$(\impliedby)$ Assume the total cost of clients from $C^*$ is equal to 0. In particular, the clients from $G$ need to have cost equal to $0$, so $H \cup H'$ must be the part of a winning committee. Similarly, the remaining $n$ facilities must correspond to the cover of $E$.

If there exists a winning committee with the total cost of clients equal to 0, then algorithm $\calA$ would find such committee. This completes the proof. 
\end{proof}

\noindent
Hence any approximation for \probowakmedian can recognize whether the instance is YES-instance. Next, we prove a more specific hardness result for the $p$-geometric sequence of weights.

\begin{theorem}\label{thm:metric_geom_no_approximation}
Consider the \probowakmedian problem for the $p$-geometric sequence of weights. For each $c < 1$, there exists no polynomial-time $(n^{-c \ln(p) - 1})$-approximation algorithm for the problem unless $\p = \np$.
\end{theorem}
\begin{proof}
Let us fix $c < c' < 1$. For the sake of contradiction, let us assume that there exists a polynomial-time $(n^{-c \ln(p) - 1})$-approximation algorithm $\calA$ for our problem. We will show that this algorithm can be used as an $c' \ln(n)$-approximation algorithm for the \probsetcover problem, which, unless $\p=\np$, will stay in contradiction with the approximation threshold established by Dinur et al.~\cite{DinurS14}.

Let us take an instance $I$ of the \probsetcover problem. In $I$ we are given a set of $n$ elements $E$, a collection $\calS$ of subsets of $E$. We ask about minimal $k$ such that there is a subcollection $C$ of $k$ subsets from $\calS$ such that each element from $E$ belongs to some set from $C$. Without loss of generality we can assume that $n$ is big enough to satisfy $c \ln(n) + 1 < c'\ln(n)$.

From $I$ we construct an instance $I'$ of \probowakmedian as follows. We set $P = \sum_{i = 1}^k p^{i-1}$ and $x = \lceil c \ln(n) +1 \rceil$, thus $c \ln(n) + 1 \leq x \leq c'\ln(n)$. For each element $e \in E$ we introduce one client. Further, for each set $S \in \calS$ we introduce $x$ facilities $S_1, \ldots, S_x$. The cost of a client $e$ for a facility $S_i$ is equal to 0 if and only if $e \in S$; otherwise, it is equal to one. Finally, we guess the optimal solution $k$ to $I$, and set the size of the desired set of facilities to $k_{\mathrm{fac}} = k \cdot x$. 

Let us assume that there exist a set cover $C$ of size $k$ for the original instance $I$. What is the cost of the clients in the optimal solution for $I'$? Let us consider the set of $k_{\mathrm{fac}}$ facilities $C'$ constructed in the following way: for each set from the cover $S \in C$ we take all $x$ facilities that correspond to $S$ and add them to $C'$. Observe that each client has cost equal to zero from at least $x$ facilities from $C'$. Thus, the total cost of clients from $C'$ is at most equal to (recall that $P = \sum_{i = 1}^K p^{i-1}$):
\begin{align*}
\wmin(C') = \sum_{j=1}^{n} \wmin(C', j) &\leq n \big( p^{x} + p^{x+1} + \ldots + p^{K-1}\big) \leq\\
&\leq n \big( p^{c \ln(n) + 1} + p^{c \ln(n)+2} + \ldots + p^{K-1}\big) < n p^{c \ln(n)} \cdot P =\\
         &= n \cdot e^{\ln(p) \cdot c \cdot \ln(n)} \cdot P = n \cdot n^{c \ln(p)} \cdot P = P n^{c \ln(p) + 1}\text{.}
\end{align*}
Now, let us take the set of $k_{\mathrm{fac}}$ facilities $C''$ such that some client has cost equal to one from each facility from $C''$. Then, $\wmin(C'') \geq P$. Thus, an $(n^{-c \ln(p) - 1})$-approximation algorithm for \probowakmedian with the $p$-geometric sequence of weights needs to return a solution where each client has cost equal to zero from at least one facility. From such solution, however, we can extract at most $k_{\mathrm{fac}}$ sets which form a cover of the original instance. Thus, our algorithm $\calA$ can be used to find $x$-approximation solutions for the \probsetcover problem (recall that $x < c'\ln(n)$), which is impossible under the standard complexity theory assumptions~\cite{DinurS14}. This completes the proof.
\end{proof}

Using that we obtain

\begin{corollary}\label{cor:metric_geom_no_approximation}
There is no polynomial-time constant-factor approximation algorithm for the \probowakmedian problem for the $p$-geometric sequence of weights when $p < \nicefrac{1}{e}$.
\end{corollary}

\end{document}